\documentclass[a4paper,aps,pra,showpacs,floatfix,superscriptaddress,twocolumn,nofootinbib]{revtex4-2}
\usepackage[osf,sc]{mathpazo}
\usepackage[colorlinks,breaklinks,linkcolor=magenta,citecolor=blue,urlcolor=blue]{hyperref}
\usepackage{amsmath,amssymb,amsthm,bm,mathtools,amsfonts,mathrsfs,bbm,dsfont}
\usepackage{physics}
\usepackage{graphicx}
\graphicspath{{figures/}}
\usepackage{color}
\usepackage{xcolor}
\usepackage{enumitem}
\usepackage{makecell}

\newcommand{\id}{\ensuremath{\mathds{1}}}

\renewcommand{\vec}[1]{\boldsymbol{#1}}

\DeclareMathOperator{\polylog}{polylog}

\newtheorem{thm}{Theorem}
\newtheorem{cor}[thm]{Corollary}

\newtheorem{res}[thm]{Result}
\newtheorem{lem}[thm]{Lemma}
\newtheorem{definition}[thm]{Definition}

\newtheorem{manuallemmainner}{Lemma}
\newenvironment{manuallemma}[1]{%
  \renewcommand\themanuallemmainner{#1}%
  \manuallemmainner
}{\endmanuallemmainner}

\newtheorem{manualresultinner}{Result}
\newenvironment{manualresult}[1]{%
  \renewcommand\themanualresultinner{#1}%
  \manualresultinner
}{\endmanualresultinner}

\begin{document}
\nonfrenchspacing

\title{Phases of Matrix Product States with Symmetric Quantum Circuits and Symmetric Measurements with Feedforward}

\author{David Gunn}
\affiliation{Institute for Theoretical Physics, University of Innsbruck, Technikerstraße 21A, 6020 Innsbruck, Austria}
\affiliation{Department of Physics, QAA, Technical University of Munich, James-Franck-Str. 1, D-85748 Garching, Germany}

\author{Georgios Styliaris}
\affiliation{Max-Planck-Institut für Quantenoptik, Hans-Kopfermann-Straße 1, D-85748 Garching, Germany}
\affiliation{Munich Center for Quantum Science and Technology (MCQST), Schellingstraße 4, D-80799 München, Germany}

\author{Tristan Kraft}
\affiliation{Institute for Theoretical Physics, University of Innsbruck, Technikerstraße 21A, 6020 Innsbruck, Austria}
\affiliation{Department of Physics, QAA, Technical University of Munich, James-Franck-Str. 1, D-85748 Garching, Germany}

\author{Barbara Kraus}
\affiliation{Department of Physics, QAA, Technical University of Munich, James-Franck-Str. 1, D-85748 Garching, Germany}
\affiliation{Institute for Theoretical Physics, University of Innsbruck, Technikerstraße 21A, 6020 Innsbruck, Austria}

\date{\today}

\begin{abstract}

Two matrix product states (MPS) are in the same phase in the presence of symmetries if they can be transformed into one another via symmetric short-depth circuits. We consider how symmetry-preserving measurements with feedforward alter the phase classification of MPS in the presence of global on-site symmetries. We demonstrate that, for all finite abelian symmetries, any two symmetric MPS belong to the same phase. We give an explicit protocol that achieves a transformation between any two phases and that uses only a depth-two symmetric circuit, two rounds of symmetric measurements, and a constant number of auxiliary systems per site. In the case of non-abelian symmetries, symmetry protection prevents one from deterministically transforming symmetry-protected topological (SPT) states to product states directly via measurements, thereby complicating the analysis. Nonetheless, we provide protocols that allow for asymptotically deterministic transformations between the trivial phase and certain SPT phases.

\end{abstract}

\maketitle

\section{Introduction\label{sec:intro}}

One of the central goals of condensed matter physics is the classification of zero-temperature quantum phases~\cite{sachdev2011quantum}.
Quantum states in that regime become pure, and thus all their correlations are due to entanglement, whose distinct patterns give rise to the different topological phases~\cite{wen2017colloquium}.
One example is the classification of topological phases for spin Hamiltonians with short-ranged interactions. Their ground states are described by tensor network states~\cite{Fannes1992,hastings2007area,hastings2006solving,molnar2015approximating,verstraete2004valencebond,buerschaper2009explicit}, a physically motivated family of states that provide a controllable entanglement structure.
The tensor network formalism not only allows for an efficient description of many-body ground states~\cite{cirac2021matrix} but has also made possible the explicit classification of topological phases in one spatial dimension, using matrix-product states (MPS)~\cite{Chen2011_Phases1,Schuch2011_Phases2}.
The landscape of phases is vastly enriched with the inclusion of symmetries, giving rise to symmetry-protected topological (SPT) phases~\cite{affleck1987rigorous,li2008entanglement,pollmann2010entanglement,Chen2011_Phases1,Schuch2011_Phases2}.
Symmetries refine the phase diagram and explain important physical features of the underlying symmetric MPS ground states, such as the degeneracy of their spectrum and the long-range entanglement of their edge modes~\cite{cirac2021matrix}.

From the point of view of quantum information, topological phases is the classification of states according to their complexity, when locality restrictions are imposed~\cite{chen2010local,Huang2015_symandlindepthcircuits}.
This operational definition can be made precise using quantum circuits ($QC$), i.e., decompositions of unitary operators into elementary gates~\cite{NielsenAndChuang}.
Two many-body states belong to the same phase if and only if they can be converted into each other by a circuit consisting of nearest-neighbor gates whose depth (i.e., number of layers) is roughly constant for all system sizes~\cite{chen2010local,Chen2011_Phases1,Huang2015_symandlindepthcircuits}.
Such unitary operations are, from a complexity perspective, simple as unitary operators generically require a depth which scales exponentially in the system size~\cite{NielsenAndChuang}.
Thus, two states corresponding to different topological phases cannot be converted into each other by this restricted class of operations. Moreover, the quantum circuit picture naturally allows for the inclusion of SPT order by further restricting the allowed gates to be symmetric~\cite{Huang2015_symandlindepthcircuits}. In quantum computing, SPT phases have been associated with the resourcefulness of a state in measurement-based quantum computation~\cite{raussendorf2001oneway,verstraete2004valencebond,stephen2017computational,Stephen2019_MBQC3}, providing a link between symmetries, measurements and computation.

Both perspectives on topological phases come together in quantum simulators, where phases of matter can be actively engineered~\cite{satzinger2021realizing,semeghini2021probing,iqbal2023topological}.
From the complexity point of view, states belonging to the trivial phase capture roughly what is feasible to prepare. This is because actual devices are heavily constrained by noise and often constrained by the locality of interactions~\cite{preskill2018quantum}; thus, low-depth circuits on product states, i.e., the trivial phase, capture these constraints.
Despite these limitations, measurements can also be performed and be used to assist state transformations.
In particular, mid-circuit measurements and subsequent conditioning of gates on the outcomes (i.e., feedforward) can, in certain situations, facilitate the conversion of states belonging to distinct topological phases~\cite{raussendorf2005longrange,aguado2008creation,bolt2016foliated,watts2019exponential,Piroli2021,tantivasadakarn2022longrange,tantivasadakarn2023hierarchy,bravyi2022adaptive,lu2022measurement,buhrman2023state}.
It is therefore relevant to consider how the classification of topological phases is altered if measurements and feedforward are included.

It is known that adding a single round of measurements to a finite-depth circuit can deterministically convert a product state to a long-range correlated state.
Examples include the GHZ state~\cite{watts2019exponential} (symmetry breaking) and the toric code in 2D~\cite{raussendorf2005longrange} (topological order). Both of these states belong to a nontrivial topological phase and collapse to the trivial one if measurements are included.
More recently, it was shown that all translation-invariant MPS collapse to the trivial phase, including long-range correlated states ~\cite{Piroli2021,malz2023preparation}.
Several other topologically ordered phases have been shown to also collapse to the trivial one, such as quantum double models~\cite{kitaev2006anyons} of finite groups~\cite{bravyi2022adaptive,lu2022measurement}, string-net models~\cite{lu2022measurement}, and a hierarchy of phases based on the number of rounds of measurements has been suggested~\cite{tantivasadakarn2023hierarchy}.

Inspired by the results of Ref.~\cite{Piroli2021}, we study how the classification of phases of MPS under symmetries is altered when symmetric measurements are allowed.
We completely resolve the case of finite abelian groups, and we find that all phases collapse to the trivial one, including those with symmetry breaking and nontrivial cohomology.
For non-abelian groups, we construct asymptotically deterministic transformations from certain SPT phases to product states, and vice versa. We also partially extend this result to non-normal phases.
However, it remains an open problem whether this is a general feature of phases with respect to non-abelian groups. A summary of our results in given in Table~\ref{tab:results}.

The rest of the paper is organized as follows. In Section~\ref{sec:prelim}, we recall some important facts about MPS, the classification of their phases with and without symmetries, as well as the role of circuits and measurements with feedforward. In Section~\ref{sec:operations}, we will define the operations that we consider and make some preliminary observations about the phase diagram that directly follow from these operations. In Section~\ref{sec:FiniteAbelianSymmetries}, we solve the classification of phases of MPS with quantum circuits and measurements with feedforward under finite abelian symmetries. In Section~\ref{sec:NonAbelianSymmetries}, we consider non-abelian symmetries and investigate normal and non-normal phases. 

\begin{table}[t]
    \centering
    \begin{tabular}{c||c|c}
         & normal & non-normal \\ \hline\hline 
        abelian  & \makecell[c]{trivializes } & \makecell[ct]{trivializes \\ (Sec.~\ref{sec:FiniteAbelianSymmetries})} \\ \hline
        non-abelian &  \makecell[ct]{Some SPT trivialize \\ (Asymptotically)\\ (Sec.~\ref{sec:SPTtrivial})} & \makecell[ct]{GHZ $\rightarrow$ trivial phase \\ (Asymptotically)\\ (Sec.~\ref{sec:GHZproduct})} \\ 
    \end{tabular}
    \caption{The results of this work. We consider the phase classification of MPS under symmetric circuits and symmetric measurements. We consider four cases: abelian symmetries vs non-abelian symmetries; normal MPS (SPTs) vs non-normal MPS (i.e., GHZ-like). Abelian-normal phases were already known to trivialize. In Sec.~\ref{sec:FiniteAbelianSymmetries}, we show non-normal phases also trivialize for all finite abelian groups. Then, we consider non-abelian symmetries. In Sec.~\ref{sec:SPTtrivial}, we construct asymptotically deterministic transformations from certain SPT phases to product states and vice versa. Finally, in Sec.~\ref{sec:GHZproduct}, we investigate non-normal MPS with non-abelian symmetries.}
    \label{tab:results}
\end{table}

\section{Preliminaries}\label{sec:prelim}

In this section, we recall some important facts about MPS and their symmetries. We also discuss the classification of 1D gapped phases of matter~\cite{Chen2011_Phases1,Schuch2011_Phases2} in the Hamiltonian and the circuit picture (with and without symmetries). Finally, we discuss the phases of MPS when measurements and feedforward are included~\cite{Piroli2021,malz2023preparation}. Throughout, we use the notation $[n]=\{0,1,\dots,n-1\}$.

\subsection{Matrix Product States and Local Symmetries}
Any translationally invariant MPS with periodic boundary conditions is defined by a single tensor, $A$, leading to the state
\begin{equation}
    \ket{\psi_n[A]}=\sum_{i_1,\dots ,i_n} \mathrm{tr}[A^{i_1}\dots A^{i_n}]\ket{i_1\dots i_n}.
\end{equation}
Here, $A^i$ is a $D\times D$ matrix for each $i\in[d]$, where $D\in\mathbbm{N}$ is the \textit{bond dimension} and $d\in\mathbbm{N}$ is the physical dimension. The definition of an MPS is robust under ``blocking" physical sites into a single super-site. Blocking $l$ sites together corresponds to mapping the tensor $A^i$ to the product $\tilde{A}^{i_1,\dots,i_l}=\prod_{j=1}^l A^{i_j}$. Any MPS tensor, after blocking sufficiently many sites, can be brought into a \emph{canonical form} (CF)~\cite{Cirac2017_FundThm1}~\footnote{Using the terminology of Ref.~\cite{Cirac2017_FundThm1}, this canonical form would be referred to as 'Block Injective Canonical Form II'. This canonical form ensures the unitarity of the gauge transformation in Eq.~\eqref{eq:FundThmInj}.}.
In this canonical form,
\begin{equation}\label{eq:TensorCF}
    A^i = \bigoplus_{\alpha\in[m]} A_\alpha^i,
\end{equation}
with the property that the $A^i$ span all matrices with the same block structure~\cite{Cirac2017_FundThm1} (with blocks not necessarily being of the same size), and where $m$ is the number of blocks. If a tensor in canonical form has only one block, then it is called \emph{injective}. An MPS is called \emph{normal} if, after blocking, its tensor in canonical form is injective. We represent the MPS tensor in Eq.~\eqref{eq:TensorCF} graphically as shown below.
\begin{equation*}
\begin{aligned}\label{eq:TensorCF2}
    \includegraphics[width=0.12\linewidth]{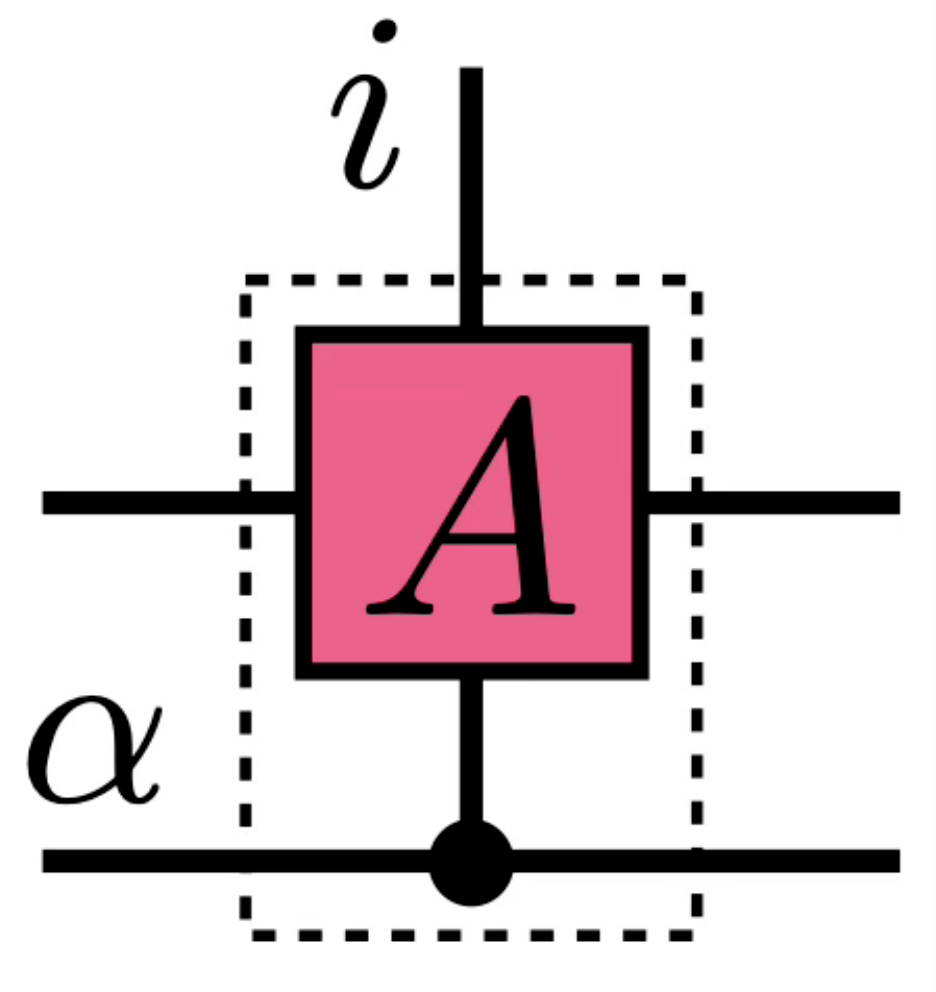}
\end{aligned}
\end{equation*}
Here, the filled circle indicates a delta function (see Ref.~\cite{Cirac2017_FundThm1} for an introduction to this graphical notation). The leg $i$ corresponds to the \textit{physical space}, whereas the horizontal legs corresponds to the so-called \textit{virtual space}. We will suppress the labeling of each leg from now on. Given a tensor $A$, one defines its so-called \emph{fiducial state} as $\ket{A}=\sum_{i,l,m}(A^i)_{lm} \ket{l}\otimes\ket{i}\otimes\ket{m}$. Graphically, the fiducial state can be obtained as shown below.
\begin{equation}\label{eq:fiducial}
\begin{aligned}
    \includegraphics[width=0.5\linewidth]{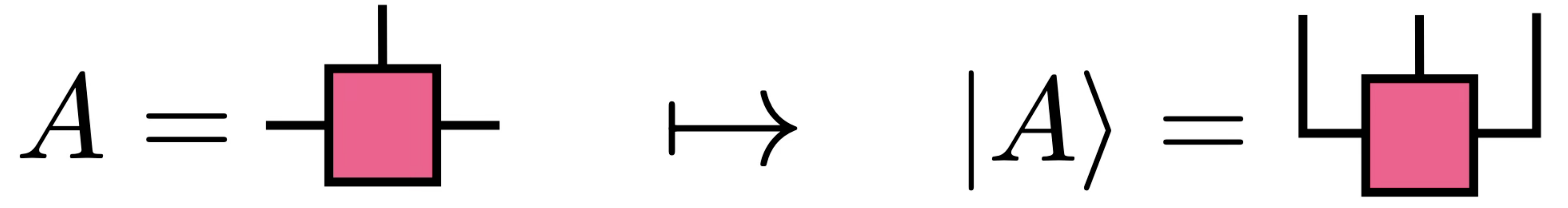}
\end{aligned}
\end{equation}
Two tensors in canonical form that yield the same MPS (up to a phase) for all $n$ are related by the so called \emph{fundamental theorem} of MPS~\cite{Cirac2017_FundThm1}. For injective MPS it states that
\begin{equation} \label{eq:FundThmInj}
    \ket*{\psi_n[A]}\propto \ket*{\psi_n [B]} \forall n \Rightarrow
     A^i = e^{i\phi} \omega^\dagger B^i \omega  ,\ \ \forall i\in[d],
\end{equation}
where $\omega$ is a unitary. For non-injective MPS the fundamental theorem states that
\begin{multline}
    \ket*{\psi_n[A]}\propto \ket*{\psi_n [B]} \forall n  \\
     \Rightarrow \bigoplus_\alpha A^i_\alpha = P^T \left( \bigoplus_\alpha e^{i\phi^\alpha}(\omega^\alpha)^\dagger   B^i_\alpha \omega^\alpha  \right) P,\ \ \forall i \in[d],
\end{multline}
where $P$ is a permutation operator that permutes the $\alpha$-blocks, $\omega_\alpha$ are unitaries, and $ e^{i\phi^\alpha}$ are phases.

The fundamental theorem can be used to characterize how local symmetries of an MPS transform its associated canonical form tensor. Let $U_{g\in G}$ be a linear unitary representation of some group, $G$. If $U_g$ is a \emph{global on-site symmetry} of an injective MPS, i.e., if $U^{\otimes n}_g\ket*{\psi_n[A]}\propto\ket*{\psi_n[A]}$ for all $g\in G$ and all $n$, then, by the fundamental theorem, for injective MPS, it must hold that \cite{SanzEtAl2009_MPSHamiltonians, Schuch2011_Phases2}
\begin{equation}\label{eq:injectivesym}
    \sum_j (U_g)_{ij} A^j = e^{i\phi_g} \omega_{g}^\dagger A^i \omega_{g} .
\end{equation}
Here, the phases $e^{i\phi_g}$ form a 1D unitary irreducible representation (irrep) of $G$. The $\omega_g$ form a \emph{projective representation} of $G$, meaning that $\omega_g\omega_h= \gamma(g,h) \omega_{gh}$, for all $g,h\in G$, where $\gamma: G\times G \rightarrow U(1)$ is referred to as the \emph{cocycle}. Crucially,  as any transformation $\omega_g\mapsto \nu(g) \omega_g$ leaves Eq.~\eqref{eq:injectivesym} invariant for $\nu_g\in U(1)$, the $\omega_g$ are only defined up to a phase. Hence $\gamma (g,h)$ is defined up to the equivalence relation $\gamma (g,h)\sim \frac{\nu(gh)}{\nu(g) \nu(h)} \gamma (g,h)$. The induced equivalence classes of the $\gamma (g,h)$ can be shown to be isomorphic to the second cohomology group of $G$ over $U(1)$, $H^2(G,U(1))$, and thus they are typically referred to as \emph{cohomology classes} of $G$ (see Appendix~\ref{app:AppendixAInjectiveMPS}). We will later see that they play an important role in the classification of phases under symmetries~\cite{Chen2011_Phases1, Schuch2011_Phases2}.

For non-injective MPS, the algebraic structure of the action of the local symmetries on the virtual level becomes more involved. A global on-site symmetry, $U_g$, acts on the tensor as~\cite{Schuch2011_Phases2}
\begin{equation}\label{eq:noninjectivesym}
    \sum_j (U_g)_{ij} \bigoplus_\alpha A^j_\alpha = P_g^T \left[ \bigoplus_\alpha e^{i\phi_g^\alpha} (\omega_{h(g,\alpha)})^\dagger A^i_\alpha \omega_{h(g,\alpha)} \right] P_g .
\end{equation}
Here, the $P_g$ form a \emph{permutation representation} of $G$ that permutes the blocks of $A$. By considering for which $g\in G$ a given block does not move, this permutation representation action fixes a subgroup $H$ of $G$. Then, $\omega_{h(g,\alpha)}$ is a projective representation of $H$, that depends on both $g$ and $\alpha$, and is associated with a cohomology class of $H$. Finally, the phases $e^{i \phi^\alpha_g}$ form a 1D irrep of $G$ for all $\alpha$ (see Ref.~\cite{Schuch2011_Phases2} for further explanation, or Appendix~\ref{app:AppendixANonInjectiveMPS}). Graphically, we may represent this as shown below.
\begin{equation}
\begin{aligned}
    \includegraphics[width=0.6\linewidth]{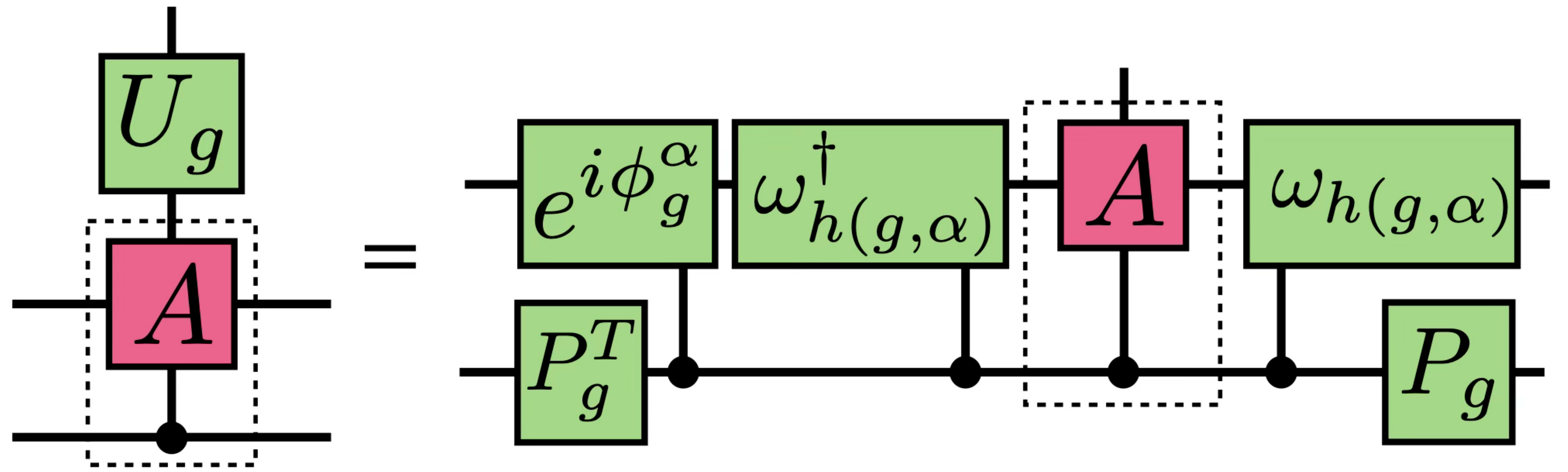}
\end{aligned}
\end{equation}

\subsection{Classification of Phases of MPS} \label{sec:ClassificationOfPhasesHamiltonian}
Any MPS in its canonical form can be associated with a \emph{parent Hamiltonian} that has this MPS as its ground state~\cite{Fannes1992,PerezGarciaEtAl2007_MPSrepresentations,Schuch2011_Phases2}. The parent Hamiltonian is a gapped, translation invariant, local Hamiltonian of the form $H=\sum_{i=1}^N h_{i,i+1}$, where $h_{i,i+1}\geq 0$. The degeneracy of the ground state space can be shown to coincide with the number of blocks in the canonical form in Eq.~\eqref{eq:TensorCF}. Importantly, the mapping between MPS and their parent Hamiltonians can be made one-to-one~\cite{PerezGarciaEtAl2007_MPSrepresentations}. Two systems are deemed to be in the same phase if their parent Hamiltonians can be transformed into one another along a path of local Hamiltonians without ever closing the spectral gap above the ground state subspace (in the thermodynamic limit).

For systems with a unique ground state, it is known that without imposing symmetries the phase diagram contains only a single phase, i.e., all normal MPS are in the same phase. Throughout, we will refer to the phase containing product states as the trivial phase. Thus, without imposing symmetries, all normal MPS are in the trivial phase. In case one imposes symmetries, the phase diagram becomes richer, leading to SPT phases. In this case, phases of normal MPS are determined by the action of the local symmetries on the virtual level, cf. Eq.~\eqref{eq:injectivesym}, and are correspondingly labeled by the elements of the second cohomology group, $H^2(G,U(1))$ \cite{Chen2011_Phases1,Schuch2011_Phases2}.

For non-normal MPS without symmetries, MPS belong to the same phase if and only if they have the same number of blocks in canonical form \cite{Schuch2011_Phases2}, i.e., the same ground state degeneracy in their parent Hamiltonian\footnote{Note, in the absence of symmetry protection, the ground state degeneracy of Hamiltonians is not stable, and thus, such Hamiltonians are often not considered to correspond to a distinct phase. See Appendix~\ref{app:noninjectivesym} for a further discussion.}. In case symmetries are imposed, phases are determined by how the symmetry permutes the blocks in the canonical form, i.e., by the action of $P_g$, and the cohomology class of $\omega_h$ in Eq.~\eqref{eq:noninjectivesym}. One can show that this corresponds to labeling phases under G by a tuple, $(H,\mu)$, where $H\le G$ is a subgroup of $G$, and $\mu\in H^2(H,U(1))$ is a cohomology class of $H$ \cite{Schuch2011_Phases2}. Such phases are referred to as symmetry breaking (see Appendix~\ref{sec:AppendixAPhasesHamiltonians} for a further discussion). An example of this phase classification is given in Appendix~\ref{app:Example}.

\subsection{Topological Phases under Quantum Circuits}\label{subsec:topophasescircuits}

Another possibility to address the classification of phases is by considering quantum circuits. Here, one defines two (sequences of) MPS $\ket{\psi_n[A]} $, $\ket{\psi_n[B]}$ to belong to the same phase if and only if there exist unitaries $U_n$ such that (i) $\| \ket{\psi_n[A]} - U_n \ket{\psi_n[B]} \| \le \epsilon$ with $\epsilon \to 0$ as $n \to \infty$, and (ii) each $U_n$ admits a decomposition as a quantum circuit with local gates and depth (i.e., number of layers) $l = O( \polylog n)$. This definition turns out to be equivalent to the Hamiltonian one~\cite{hastings2005quasi,osborne2006efficient,CoserandGarcia2019_PhasesMixedStatesDissipativeEvo,haah2021quantum, Schuch2011_Phases2, malz2023preparation, Huang2015_symandlindepthcircuits}.  Ref.~\cite{Huang2015_symandlindepthcircuits} argues that if the gates comprising the quantum circuit are also restricted to be symmetric, then, in the absence of symmetry breaking, one obtains phases under symmetries as in the Hamiltonian picture outlined in the previous section (see also Ref.~\cite{CoserandGarcia2019_PhasesMixedStatesDissipativeEvo}). We comment on the connection between the Hamiltonian and circuit picture in more detail in Appendix~\ref{sec:AppendixAPhasesCircuits}.

The landscape of the different topological phases thus exactly consists of the classes of states among which transformations (in the thermodynamic limit) are impossible under (symmetric) local quantum circuits of constrained depth. These operations, however, do not allow for measurements.

\subsection{Topological Phases with Circuits, Measurements and Feedforward (CMF)}\label{subsec:topophasescircuitsandfeedforward}

To incorporate measurements, additional auxiliary systems are usually allowed, which interact locally with the rest of the system, can be measured, and also traced out (i.e., not considered as output)\footnote{Note, the classification of phases of MPS according to Ref.~\cite{Schuch2011_Phases2} (with symmetries) is unchanged if one also has access to (symmetric) auxiliary systems.}. Importantly, gates that come after measurements can be classically conditioned on the previous measurement outcomes. Such operations are referred to as circuits with measurements and feedforward (CMF). Note that the classical conditioning (i.e., classical communication) does not need to obey any locality restrictions.

In the context of MPS without any symmetry constraints, it was shown in Refs~\cite{Piroli2021,malz2023preparation} that the inclusion of measurements and feedforward collapses the phase diagram to a single phase. In particular, all translation-invariant MPS over $n$ sites with a constant bond dimension (including the non-normal ones with long-range entanglement) are deterministically reachable from the product state, either (i) by circuit depth $O(\log n)$ and a single round of measurements, or (ii) by $O(\log \log n)$ circuit depth and the same number of rounds of measurements. In both cases, the transformation is approximate but the error vanishes in the thermodynamic limit. In this paper, we study the phases of matter that follow from considering transformations via circuits with measurement and feedforward that both respect a given symmetry. A summary comparing the phase classification with and without measurements and/or symmetries is provided in the Table \ref{tab:my_label}.

\begin{table}[t]
    \begin{tabular}{c||c|c}
         & \makecell[ct]{Gapped Hamiltonians/ \\Circuits}   &\makecell[ct]{Circuits, Measurements \\
         $\&$ Feedforward}  \\ \hline\hline  
        No Sym.  &  \makecell[ct]{\# blocks in\\  canonical form~\cite{Schuch2011_Phases2}}& \makecell[ct]{all MPS trivial\\ \cite{Piroli2021, malz2023preparation}} \\ \hline
        \makecell[ct]{Sym.\\ $(G, U_g)$} & \makecell[ct]{permutation action and \\ cohom. classes~\cite{Chen2011_Phases1, Schuch2011_Phases2}}  & \makecell[ct]{This work.\\ See also \cite{CoserandGarcia2019_PhasesMixedStatesDissipativeEvo, deGroot2021_SPTInOpenQSystems}}\\
    \end{tabular}
    \caption{Characterization of the phases of MPS under different operations. Two MPS are in the same phase (with symmetries) if they can be connected by a path of (symmetric) gapped parent Hamiltonians. This classification is equivalent to the classification under (symmetric) short-depth circuits~\cite{hastings2005quasi,osborne2006efficient,CoserandGarcia2019_PhasesMixedStatesDissipativeEvo,haah2021quantum} (cf. Section~\ref{subsec:topophasescircuits}). In the Hamiltonian/circuit classification without symmetries, all MPS with the same number of blocks in the canonical form [see Eq.~\eqref{eq:TensorCF}] belong to the same phase \cite{Schuch2011_Phases2}. By imposing symmetries, the phase classification is enriched as now phases depend on the permutation action and the group cohomology induced by the action of the symmetry on the MPS tensor [see Eq.~\eqref{eq:noninjectivesym}] ~\cite{Schuch2011_Phases2,Chen2011_Phases1}. In the case where measurements with feedforward are added to the circuit classification, the phase diagram trivializes if no symmetries are imposed~\cite{Piroli2021,malz2023preparation}. Here, we investigate how this classification is altered if one imposes symmetries.}
    \label{tab:my_label}
\end{table}

\section{Symmetric Quantum Circuits and Symmetric Measurements with Feedforward}\label{sec:operations}

In this section we introduce a symmetry-preserving extension of CMF, which we denote by G-CMF. We also will discuss the role of auxiliary systems in assisting transformations, and we will make some preliminary observations on possible transformation between phases that follow directly from the definition of G-CMF operations.

\subsection{Definition of G-CMF Operations}

First, one fixes a linear unitary representation $U_{g\in G}$ of a group $G$. A unitary, $V$, on $m$ qudits is called symmetric if it commutes with $m$ copies of the representation $U_g$, i.e., if $[V,U_g^{\otimes m}]=0$ for all $g\in G$. Similarly, a projective measurement on $m$ qudits, $\{P_k\}_k$, with $P_k^2=P_k\geq 0$, and $\sum_k P_k=\id$, is symmetric if $[P_k,U_g^{\otimes m}]=0$, for all $g\in G$ and for all $k$. G-CMF operations are then defined by sequences of (i) applying circuits where each gate acts on $O(1)$ geometrically-local physical sites, (ii) appending or removing $O(1)$ auxiliary systems on-site in an unentangled, pure, symmetric state, and (iii) performing on-site symmetric projective measurements. The operations are summarized in Fig.~\ref{fig:GQCccl}. We can establish the following notion of equivalence of MPS under G-CMF operations.

\begin{definition}[G-CMF equivalence] \label{def:G-CMF}
     Let $G$ be a group and $U_g$ a linear unitary representation thereof. We say that two MPS, $\ket{\psi_n[A]} $ and $ \ket{\psi_n[B]}$, with global on-site symmetry $U_g$, are in the same phase under $G$ if they can be transformed (in the sense of convergence of sequences of MPS explained in Section~\ref{subsec:topophasescircuits}) into one another via G-CMF operations with a cumulative $O(\text{polylog}(n))$ depth circuit.
\end{definition}

Before moving on, let us discuss some limiting cases of our definition:
\begin{enumerate}
    \item In the absence of symmetries, our definition collapses to the operation $QCcc_{O(\polylog n)}$ as introduced in Ref.~\cite{Piroli2021} (with the additional constraint of only $O(1)$ auxiliary systems per physical site, whereas in $QCcc_l$ the number of auxiliary systems is in principle unbounded).
    \item In the absences of measurements, our definition collapses to phases of MPS as classified by symmetric quantum circuits (or equivalently, as discussed above, symmetric gapped Hamiltonians).
    \item In the absence of circuits, our definition corresponds to a symmetric version of Local Operations and Classical Communication (LOCC) \cite{Chitambar2014_EverythingLOCC}, with additional constraints on the auxiliary systems (see Refs \cite{SchuchEtAl2004_SuperselectionLOCC2, deGroot2020_InaccessEnt,deGroot2021_SPTInOpenQSystems} for similar operations).
\end{enumerate}

Next, let us explain the motivation for defining the operations as explained above.

\begin{figure}[t]
        \includegraphics[width=0.8\linewidth]{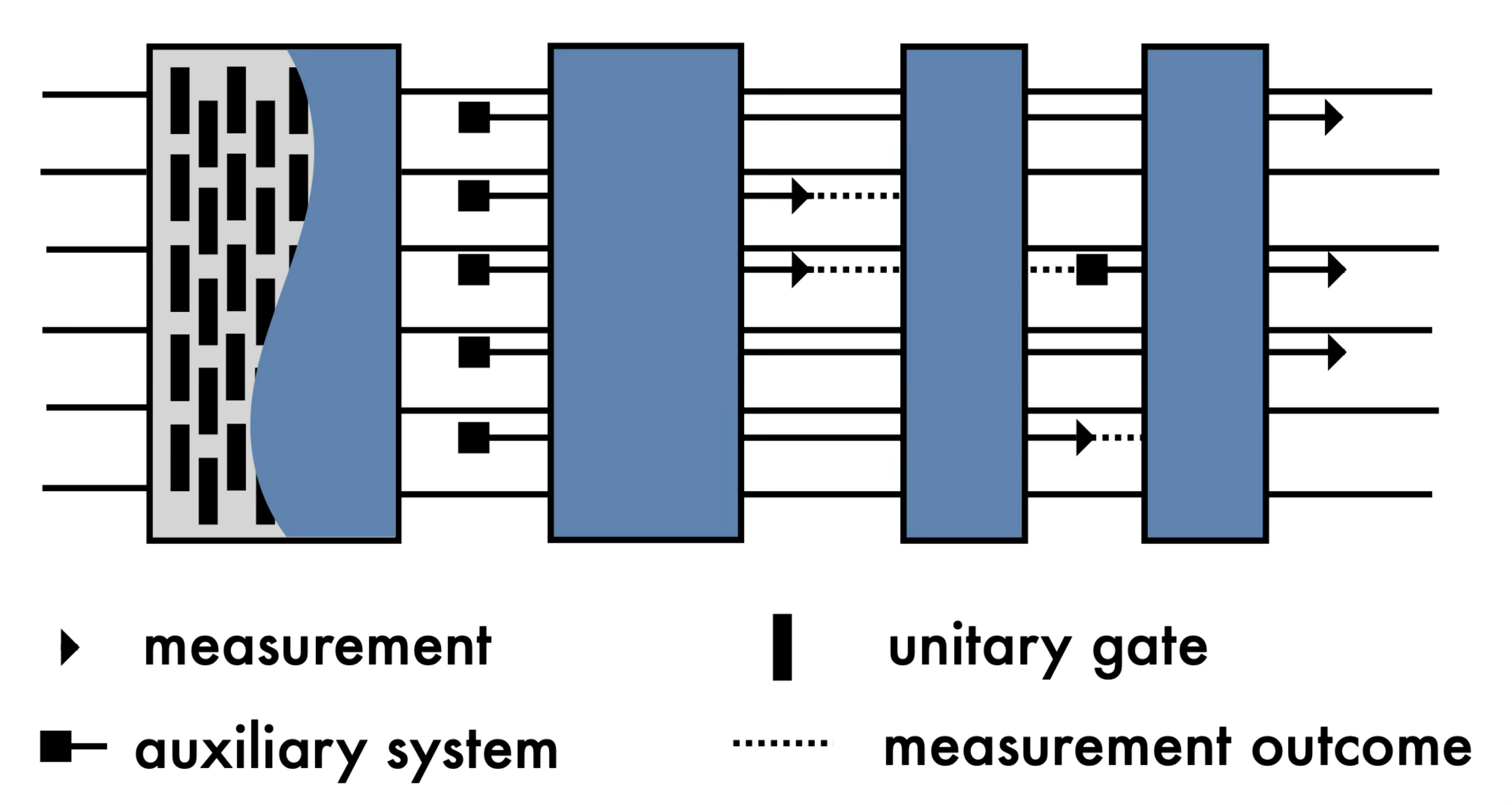}
        \caption{G-CMF operations consist of sequences of applying a local quantum circuit, where each gate acts on $O(1)$ physical sites, appending (or discarding) $O(1)$ auxiliary systems in an unentangled, pure, symmetric state, and on-site symmetric measurements with feedforward. Two states are in the same phase if they can be transformed into one another with these operations, such that the total circuit depth is $O(\text{polylog}(n))$}
        \label{fig:GQCccl}
\end{figure}

\subsection{Motivation for the Constraints on Auxiliary Systems} \label{sec:ancillas}
The constraints one imposes when adding and discarding auxiliary systems has considerable consequences for the classification. Here, we only allow auxiliary systems to be added and discarded on-site and only if they are in a pure, common eigenstate of $U_g$. Firstly, this ensures that adding and discarding auxiliary systems cannot lead to a non-symmetric state. Secondly, if one were allowed to discard entangled states, one could transform any phase to the trivial phase, independent of the symmetry group, by simply discarding the initial state.
Moreover, Eq.~\eqref{eq:TrToSPT} would hold, due to a similar argument, but independent of the symmetry.

Finally, if we were allowed to discard mixed, separable states, then one could transform any state to the trivial phases by applying an entanglement breaking channel locally and then discarding it. By only allowing auxiliary systems to be discarded if they are in a pure, common eigenstate of $U_g$, we avoid this possibility. As a final comment, our definition allows one to consider auxiliary systems as similar to catalysts: a common eigenstate can be appended locally to assist the transformation, but a common eigenstate must be returned at the end of the transformation.

Another important aspect of using auxiliary systems is that it allows us to implement quasi-commuting unitaries, while still satisfying the stronger condition of commuting with the symmetry. To see this, consider a unitary, $V$, that quasi-commutes with $U_g$; that is, $U_g V = e^{i\phi_g} V U_g,\ \forall g\in G$. It is easy to verify that $e^{i\phi_g}$ must be a 1D irrep of $G$. Moreover, using Schur's Lemma and the Burnside-Brauer theorem~\cite{Isaacs2003-no}, one can show that the fact that $V$ quasi-commutes means there is an $m\in\mathbbm{N}$ such that $U_g^{\otimes m}$ contains $\ketbra{\phi_0}+e^{-i\phi_g}\ketbra{\phi_1}$ as a sub-representation, where $\ket{\phi_{0}},\ket{\phi_{1}}$ are two orthonormal vectors. Moreover, it is easy to verify that the unitary given by $\tilde{V}=(\ketbra{\phi_1}{\phi_0}\otimes V + h.c.)\oplus \id$ commutes with $U_g^{\otimes m}\otimes U_g$ and  $\tilde{V}(\ketbra{\phi_0} \otimes \rho) \tilde{V}^\dagger  = \ketbra{\phi_1} \otimes V \rho V^\dagger$ for all $\rho\ge0$. Thus, after appending $m$ auxiliary systems, initialized in the symmetric state $\ket{\phi_0}$ and applying the strictly commuting $\tilde{V}$, the auxiliary systems are left in the state $\ket{\phi_1}$ and thus may be discarded, and thus one will have implemented $V$. 

Next, we will make some observations that follow directly from our discussion above.

\subsection{Initial Observations}\label{subsec:initial}

Under abelian symmetries, a transformation from any nontrivial phase to the trivial phase is always possible, i.e.,
\begin{equation}
    \ket{\psi}\mapsto\ket{\rm Triv}.
\end{equation}
This transformation can be achieved by measuring each physical system on-site in the common symmetric eigenbasis of the representation $U_g$. Such a basis exists as any representation of an abelian group can be decomposed into 1D irreps. This transforms the state to a symmetric pure product state, which can be discarded and replaced by a symmetric state in the trivial phase.

Additionally, it has been observed (see, e.g., Ref.~\cite{CoserandGarcia2019_PhasesMixedStatesDissipativeEvo}) that, for abelian symmetries, one can always transform from the trivial phase to any SPT phase, i.e.,
\begin{equation}\label{eq:TrToSPT}
    \ket{\rm Triv}\mapsto\ket{\rm SPT}.
\end{equation}
To do so, one notices that for any SPT state, $\ket{\text{SPT}(\mu)}$, belonging to cohomology class $\mu\in H^2(G,U(1))$, there is another SPT state, $\ket{\text{SPT}(\mu^{-1})}$, such that the state $\ket{\text{SPT}(\mu)}\otimes \ket{\text{SPT}(\mu^{-1})}$ is in the trivial phase~\cite{Schuch2011_Phases2}. Thus, one can transform $\ket{\text{Triv}}$ to $\ket{\text{SPT}(\mu)}$ by first transforming to $\ket{\text{SPT}(\mu)}\otimes \ket{\text{SPT}(\mu^{-1})}$ and then transforming the state $\ket{\text{SPT}(\mu^{-1})}$ to the trivial state by the procedure outlined above. Thus, normal MPS all belong to the same phase under abelian symmetries.

This construction does not work for non-normal MPS. Moreover, in the case of non-abelian symmetries, one cannot projectively measure onto a common symmetric eigenbasis of $U_g$ as no such basis exists. Thus, in the subsequent sections, we investigate the phases of MPS in these two directions. A summary of our results is provided in Table~\ref{tab:results}.

\subsection{The Role of the Global On-Site Symmetry} \label{sec:OnSiteSymmetryPhases}
In defining our operations, we have chosen to consider a fixed physical system with a fixed global on-site symmetry $U_g$. We do this because it fits more naturally when defining a symmetry-preserving form of circuits, auxiliary systems, and measurements. However, the classification of phases under symmetries in Refs~\cite{Schuch2011_Phases2,Chen2011_Phases1} relies only on group properties of $G$ and its subgroups (cohomology classes are a group property). Consequently, one might reasonably wonder how the on-site, physical symmetry plays a role in this classification. The resolution to this tension lies in the fact that, what determines the phase of an MPS is the action of the symmetry in the \textit{virtual space} of the tensor (after blocking). A more detailed discussion of this technicality is discussed in Appendix~\ref{sec:BlockingAndLocalSymmetry}.

\section{Finite Abelian Symmetries}\label{sec:FiniteAbelianSymmetries}

In this section, we show that the entire phase diagram under finite abelian groups trivializes under G-CMF. As discussed above, it is clear that SPT phases will trivialize under G-CMF for abelian groups. Here we show that the symmetry breaking phases, i.e., non-normal MPS, also trivialize. That is, we prove our first main result.
\begin{res} \label{res:AbelianPhaseDiagram}
Let $G$ be a finite abelian group. Then all MPS are in the same phase under G-CMF.
\end{res} 

To show this, we consider two states in different phases, $\ket{\psi(H_0,\mu_0)}$ and $\ket{\psi(H_1,\mu_1)}$. We will show that these states can be transformed into each other with G-CMF operations. In particular, for every phase, $(H,\mu)$, we construct a representative state, $\ket{\tilde{\psi}(H,\mu)}$, that is reachable from a product state, $\ket{\text{Triv}}$. The claim then follows as one can implement the sequence of symmetry-preserving transformations
\begin{multline}
    \ket{\psi(H_0,\mu_0)} \mapsto \ket{\text{Triv}}   \mapsto\ket{\tilde{\psi}(H_1,\mu_1)} \mapsto \ket{\psi(H_1,\mu_1)}.
\end{multline}
The first transformation is possible as one can always transform to the product state when $G$ is abelian (see Section \ref{subsec:initial}), and the last transformation is possible as $\ket{\tilde{\psi}(H_1,\mu_1)}$ and $\ket{\psi(H_1,\mu_1)}$ belong to the same phase. Thus, all that remains is to show that the representative state $\ket{\tilde{\psi}(H,\mu)}$ can be reached from a product state. 

\begin{figure}[t]
        \includegraphics[width=.7\linewidth]{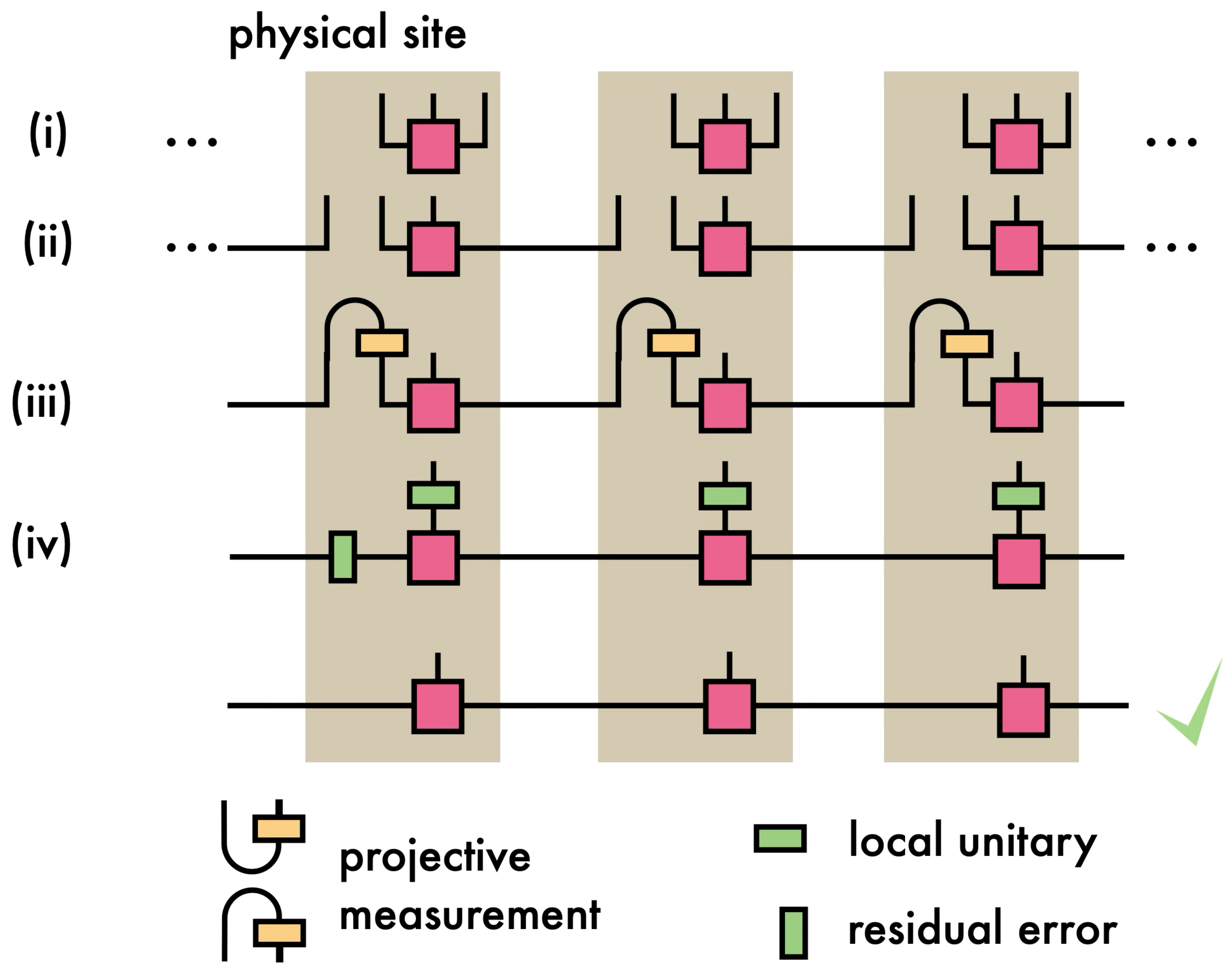}
        \caption{We consider protocols of the following form: (i) Starting from a target MPS, one locally prepares its fiducial state. (ii) Then one uses a shift circuit to move one qudit of the fiducial state to the next-nearest neighbor. (iii) Then, one perform a projective measurement on two of the qudits locally on each site. (iv) If the measurements are chosen properly, depending on the measurement results, one arrives at states that are each local unitary equivalent to the target MPS, up to some residual error. The rest of the paper is then concerned with ensuring all these steps can be performed in a symmetry-preserving way and to ensure that any residual error that cannot be corrected on the physical site occurs with zero probability.}
        \label{fig:Protocol}
\end{figure}

To show this, we use a protocol consisting of four steps, all depicted in Fig.~\ref{fig:Protocol}. Step (i): Given a tensor, $A$, of our target MPS, one locally prepares its fiducial state, defined in Eq.~\eqref{eq:fiducial}. Step (ii): We use a depth-2 circuit to move the right qudit of the fiducial state to the next neighbor. Step (iii): We locally perform a symmetric "generalized Bell measurement" [defined below in Eq.~\eqref{eq:projMeasurement}]. Depending on the measurement outcome, one obtains different post-measurement states. As we will see later, all these states are, up to on-site quasi-symmetric unitaries, equivalent to the target MPS. Importantly, we will also see that potential residual errors that cannot be corrected on the physical systems do not occur in our protocol. Step (iv): In a final step, we correct the post-measurement states to reach the desired final state.

In what follows, we will show that all the steps described above can be performed while respecting the symmetry. An example to illustrate our protocol is given in Appendix~\ref{app:Example} for the group $\mathds{Z}_4\times\mathds{Z}_2$. The detailed proof that the protocol outlined above is indeed symmetry-preserving is provided in Appendix~\ref{app:abelian}. We outline here the main ideas of why this construction works. 

We begin by choosing\footnote{Note, in Def.~\ref{def:G-CMF}, we fix a physical symmetry, whilst here we choose a physical symmetry. Indeed, we have abused notation a little bit. Strictly speaking, the RHS of Eq.~\eqref{eq:Ug} appears as a sub-representation of some finite tensor power of the original global on-site symmetry, i.e., $U_g^{\otimes m}$ for some $m\in\mathbbm{N}$. The fact that the target phase, $(H,\mu)$, is non-empty ensures such an $m\in\mathbbm{N}$ exists (see Appendix~\ref{sec:BlockingAndLocalSymmetry} for further details). As our operations are all $O(1)$ local, we are allowed to consider our state after blocking $m$ sites together. Thus we can "choose" the symmetry in Eq.~\eqref{eq:Ug}. Likewise, the state defined by the tensor in Eq.~\eqref{fig::PhaseRepresentativeTensor} should also be considered at length scale $m$. That is, we consider $\ket{\psi_{\tilde{n}}[A]}\in \mathcal{H}^{\otimes n=m\tilde{n}}$. Generally, we work at this length scale throughout the following derivation. As $m$ is finite, all symmetric operations at this length scale can be implemented in finite-depth by the symmetric operations at the original length scale (with $O(1)$ auxiliary systems locally).} the action of the global on-site symmetry, $U_g$, on the physical space as
\begin{equation}\label{eq:Ug}
     U_g = \left[\bigoplus_{\alpha\in K}  (\omega^\mu)_{h(g,\alpha)}^*\otimes (\omega^\mu)_{h(g,\alpha)}\right] \left(P_{k(g)} \otimes \id\otimes \id\right).
 \end{equation}
Here, we choose the permutation representation to be of the form $P_{k(g)}\equiv X^{k(g)}_{|K|}$, where $X$ are cyclic permutations labeled by a function, $k(g)$, which maps elements of $G$ to the quotient group $K=G/H$. By $\omega^\mu_h$ we denote an irreducible projective representation of $H$ corresponding to the cohomology class $\mu\in H^2(H,U(1))$. Graphically, we may write this as shown below.
\begin{equation}
\begin{aligned}
    \includegraphics[width=0.44\linewidth]{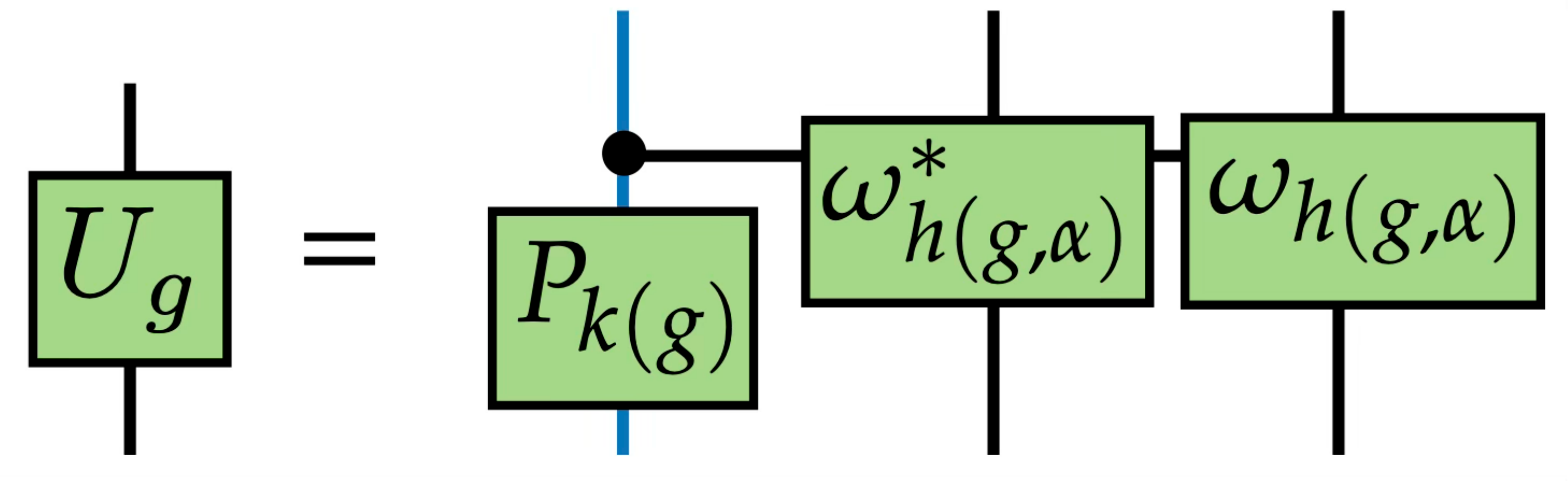}
\end{aligned}
\end{equation}
Having defined the symmetry in this way, it is easy to choose a representative state, $\ket{\tilde{\psi}(H,\mu)}$, for the phase $(H,\mu)$. To this end, we consider the following tensor.
\begin{equation}
\begin{aligned}\label{fig::PhaseRepresentativeTensor}
    \includegraphics[width=0.32\linewidth]{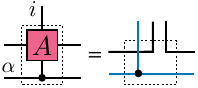}
\end{aligned}
\end{equation}
Note that in Eq.~\eqref{fig::PhaseRepresentativeTensor}, the blue lines are decoupled from the black lines. Thus, this tensor generates the MPS consisting of a tensor product of a GHZ state (generated by the blue lines) and a chain of bell pairs (generates by the black lines). By construction, this state belongs to the phase $(H,\mu)$. Having chosen our representative, let us now verify each step of the protocol is symmetric.

One can easily verify that the following (fiducial) state is symmetric with respect to $U_g^{\otimes 3}$. We note that it is almost the fiducial state of the representative tensor, but contains an additional Bell state to make it symmetric.
\begin{equation}\label{eq:FiducialRep}
\begin{aligned}
    \includegraphics[width=0.54\linewidth]{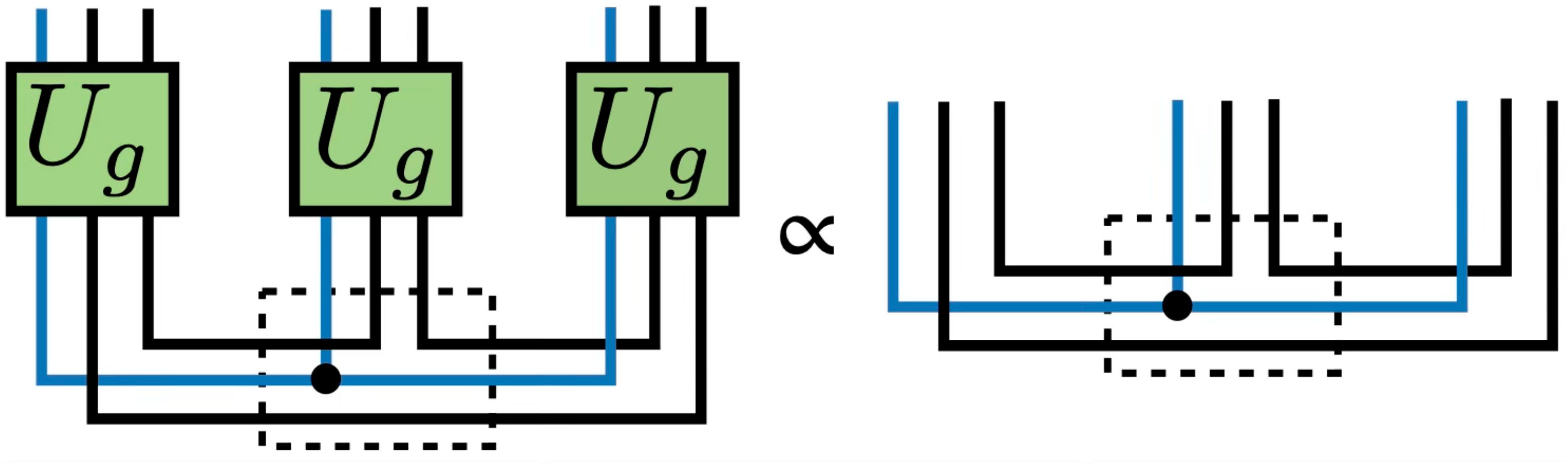}
\end{aligned}
\end{equation}
The depth-2 circuit that moves the right qudit to the nearest neighbor is symmetry-preserving as it simply swaps qudits between nearest neighbors. Such gates clearly commute with $U_g^{\otimes 2}$.

Next, let us consider the on-site measurement on the two qudits in step (iii) of the protocol, see Figure~\ref{fig:Protocol}. It corresponds to a generalized Bell measurement, with the measurement operators given by
\begin{align}\label{eq:projMeasurement}
     P_{r,q} &= (\id \otimes V_{r,q}) \ketbra*{\widetilde{\Phi}^+}(\id \otimes V_{r,q}^\dagger),
\end{align}
where
\begin{equation}
    \ket*{\widetilde{\Phi}^+}\propto\sum_{i\in [|K|]} \sum_{j,k\in [D_\mu]} \ket{i,j,k,i,k,j}
\end{equation}
is a maximally entangled state, $\abs{K}=\abs{G}/\abs{H}$, and $D_\mu$ is the dimension of the projective representation of $H$, $\omega^\mu_h$. Graphically, this may be represented as follows.
\begin{equation}\label{eq:measureMain}
\begin{aligned}
    \includegraphics[width=0.31\linewidth]{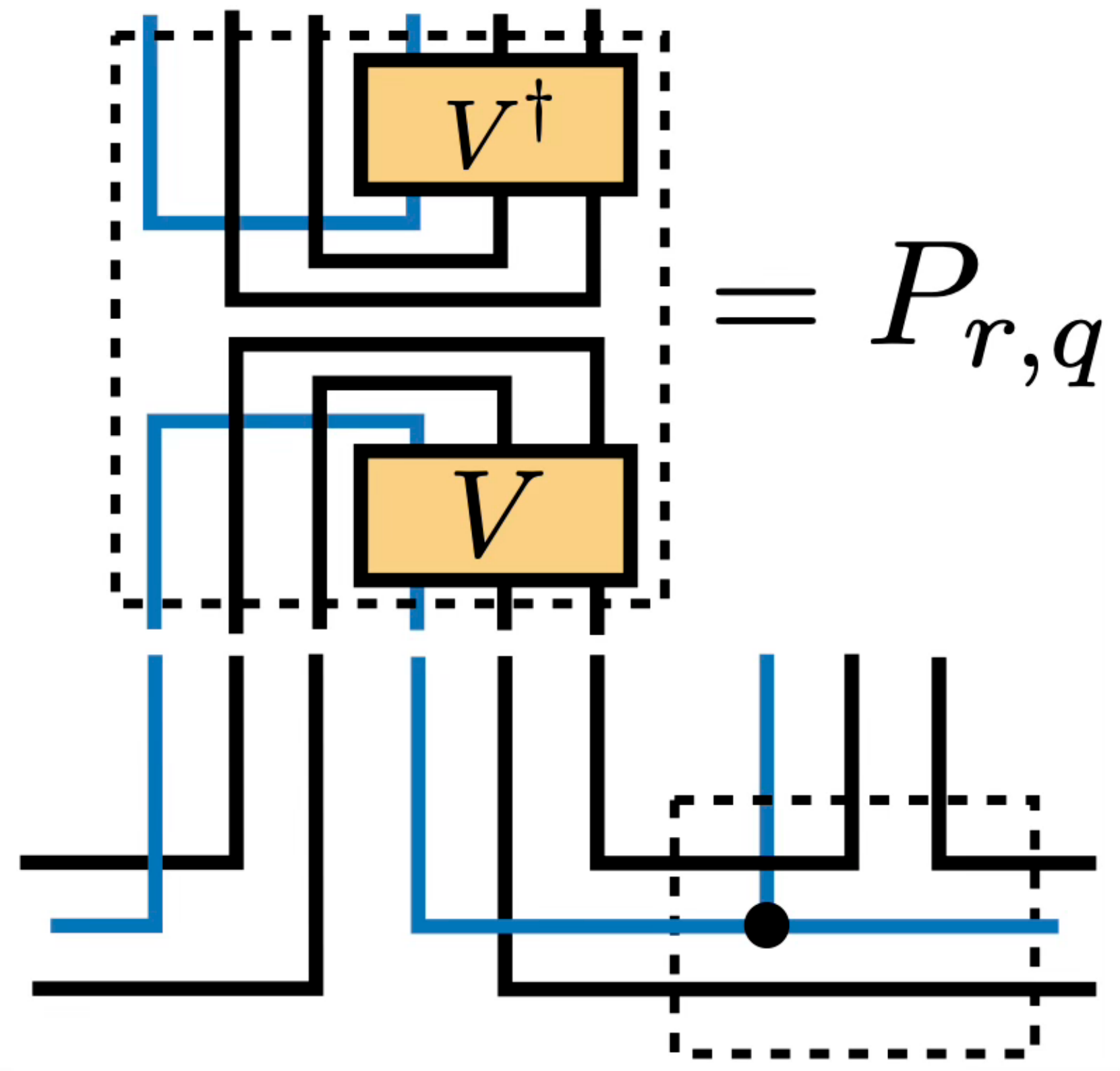}
\end{aligned}
\end{equation}

The unitaries, $V_{r,q}$, are given by
\begin{equation}\label{eq:Projectors}
    V_{r,q} = U_r \tilde{V}_q\equiv U_r\  [\tilde{Z}^q \otimes  \id \otimes (\omega_{h(q)}^\mu)],
\end{equation}
where $U_r$ is the physical symmetry, $\omega_h^\mu$ is the projective representation of $H$, and $\tilde{Z}^q$ is a diagonal unitary matrix, with $r,g\in S^\mu$, for a specific subset $S^\mu \subseteq G$ (see Appendix~\ref{app:abelian} for further details). 

Let us explain our construction, $P_{r,q}$, in more detail. We need to ensure that the measurement has the following properties: (i) The projectors, $P_{r,q}$, must commute with all the $U_g$, which is equivalent to $V_{r,q}$ quasi-commuting with all the $U_g$, i.e., $V_{r,q}U_g\propto U_gV_{r,q}$. (ii) The measurement must be complete. (iii) All post-measurement states must be correctable to the target state.

Clearly, the symmetries, $U_g$, themselves fulfill all these properties. However, there are not sufficiently many symmetries to complete the measurement, i.e., for $\tilde{V}_{\rm id}=\id$ in Eq.~\eqref{eq:Projectors}, the corresponding projectors $\{P_{r,{\rm \id}}\}_{r\in G}$ do not form a complete measurement. Thus, we have to choose additional unitaries. We choose these to be product operators so that they propagate nicely through the tensor. The first tensor factor, $\tilde{Z}^q$, of $\tilde{V}_{q}$ is a product of a generalized Z gate and an additional phase matrix. The generalized Z gate quasi-commutes with the permutation action, a generalized $X$ gate, of $U_g$ and, together with $X$, generates orthogonal operators. The third tensor factor carries a projective representation. It quasi-commutes with all the projective representations in all the blocks of $U_g$ individually, but generates potentially different phases in each of the blocks. The reason for the additional phase matrix on the first tensor factor is then precisely to cancel these phases and make the operator overall quasi-commute with the symmetry. To summarize, we have the following Lemma, which we prove and state in more detail in Appendix~\ref{app:LemmaBell}.
\begin{lem}[informal]
    \label{lem:bellmeasurements}
  Let $G$ be a finite abelian group and $(H,\mu)$ label a phase. Let $U_g$ be given as in Eq.~\eqref{eq:Ug}, and let $\{P_{r,q}\}_{r,q\in S}$ be given as in Eq.~\eqref{eq:projMeasurement}. Then, $\{P_{r,q}\}_{r,q\in S}$ is a complete, symmetric projective measurement.
\end{lem}
Finally, we must verify that we can indeed correct any post-measurement state to the target state that we wish to reach. To that end, one notices that the unitaries, $U_r$ and $\tilde{V}_q$, can be propagated though the post-measurement state according to the following rules (see Appendix~\ref{sec:AbelianSlideThrough} for a proof). 
\begin{equation}\label{eq:slidethrough}
\begin{aligned}
    \includegraphics[width=0.45\linewidth]{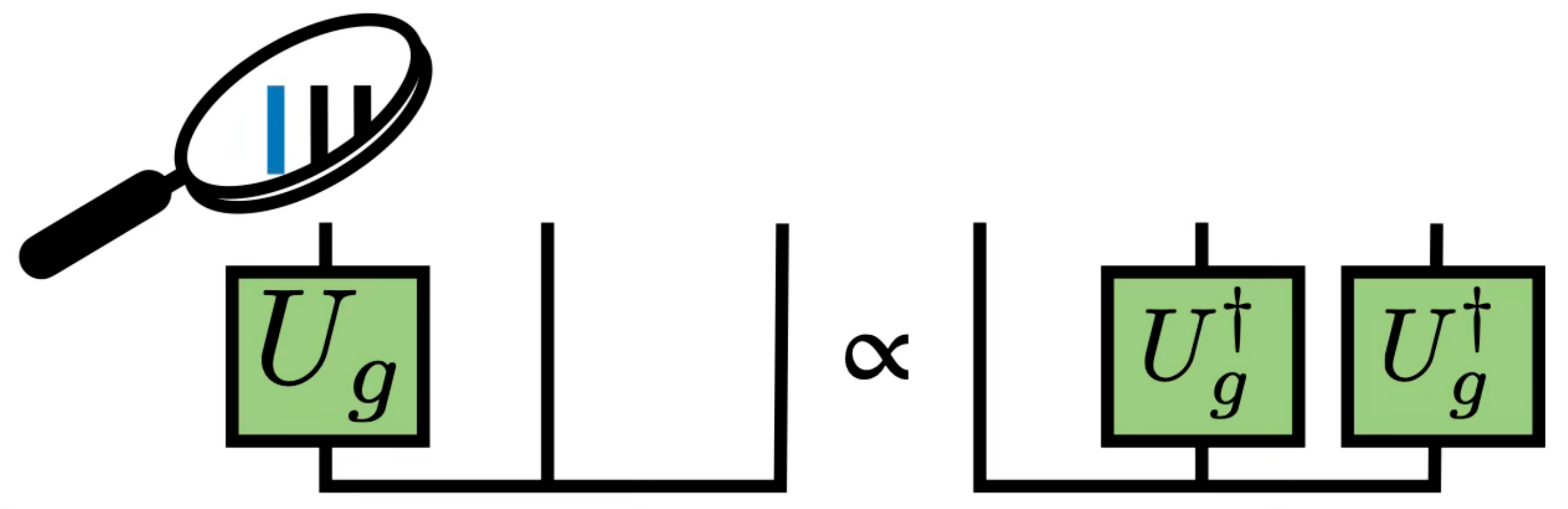}\hspace{0.06\linewidth}\includegraphics[width=0.45\linewidth]{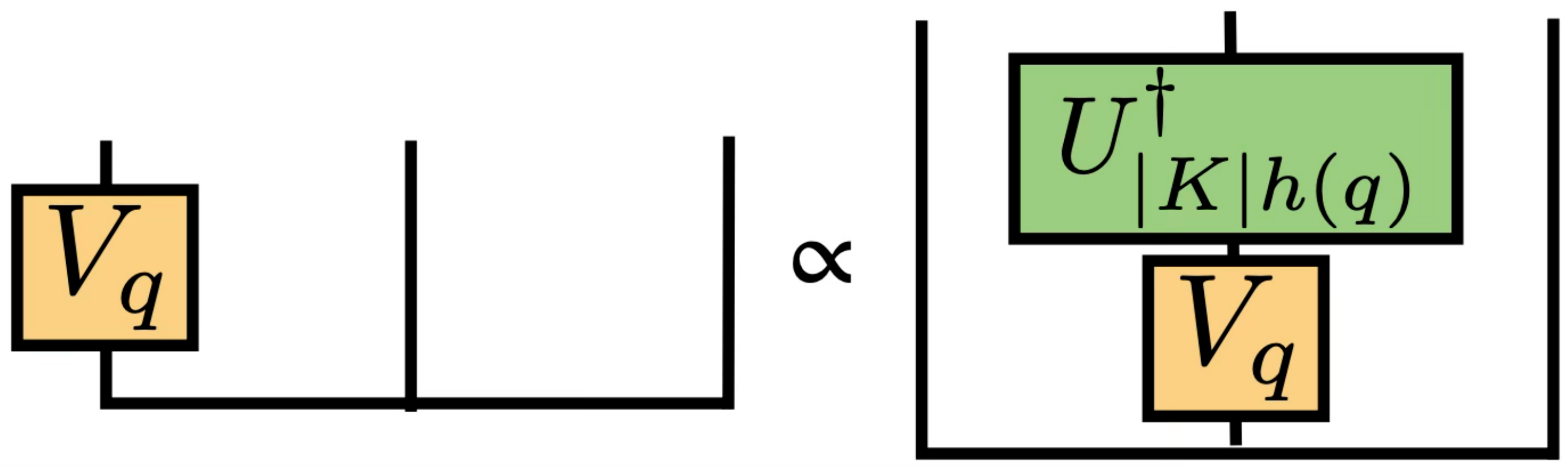}
\end{aligned}
\end{equation}
Thus, if all parties locally measure $\{P_{r,q}\}_{r,q\in S}$, yielding the outcome $(r_1,q_1,r_2,q_2,\dots,r_n,q_n)$, we can propagate the corresponding operators $V_{r,q}$ through the tensor according to the rules above and obtain the following post-measurement states.
\begin{equation*}
\begin{aligned}
\includegraphics[width=0.4\linewidth]{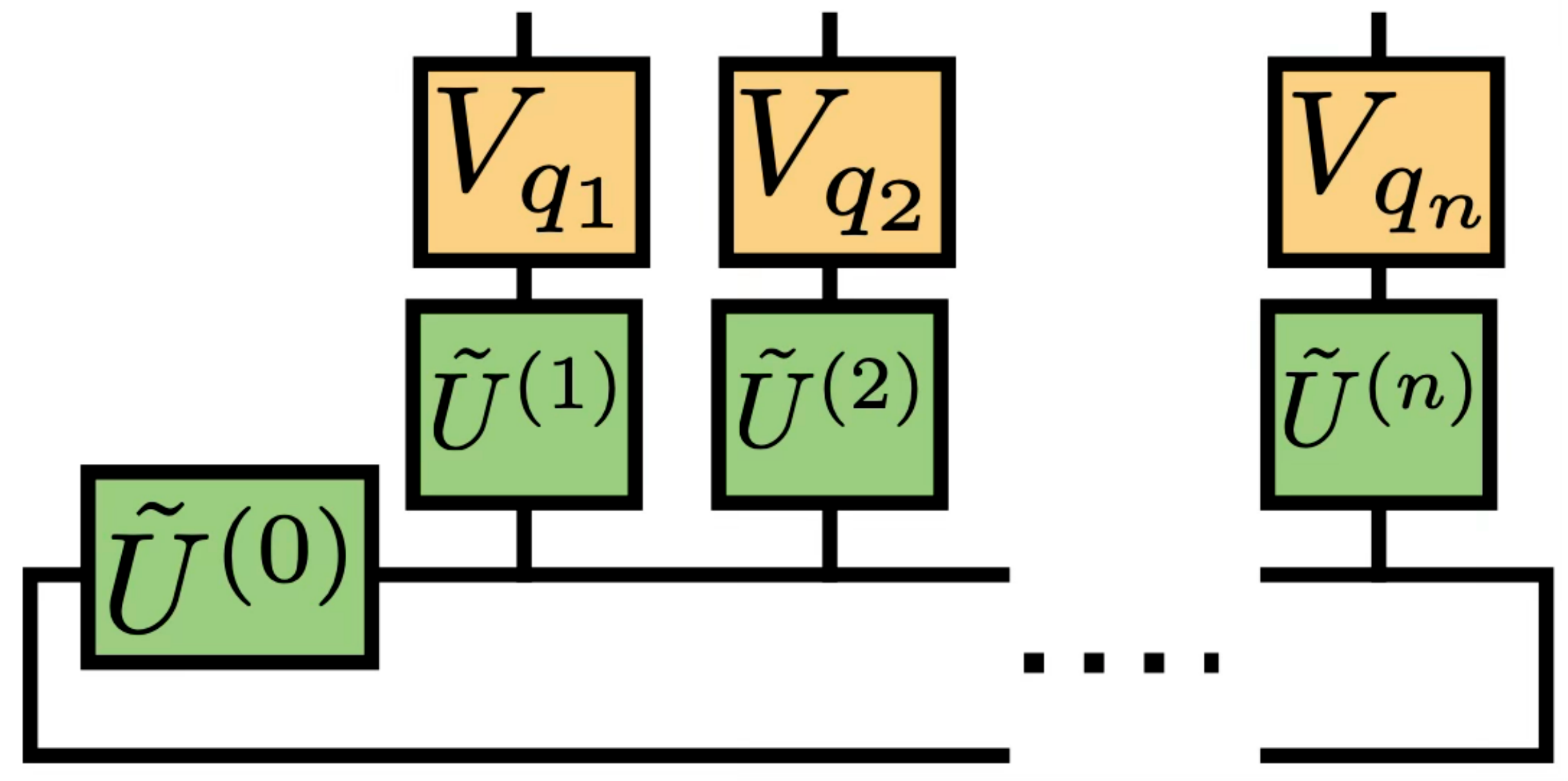} 
\end{aligned}
\end{equation*}
Here, $\tilde{U}^{(i)} = ( \prod_{j=1}^i U_{r_j} ^\dagger) U_{|K|h(q_i)}^\dagger \in \{U_g\}_{g\in G}$, and $\tilde{U}^{(0)} = \prod_{j=1}^n U_{r_j}$. If $\tilde{U}^{(0)} = \id $, then this output state is, up to local, quasi-commuting unitaries, equivalent to our desired state. Otherwise, we have the following Lemma, which we prove in Appendix~\ref{app:ZeroVectorOutcome}, ensuring that all other cases have vanishing probability.
\begin{lem}\label{lem:ZeroVectorOutcome}
    Let $\tilde{U}^{(0)} \ne \id$. Then the following equation holds.
    \begin{equation}\label{eq:ZeroVectorOutcome}
    \begin{aligned}
    \includegraphics[width=0.4\linewidth]{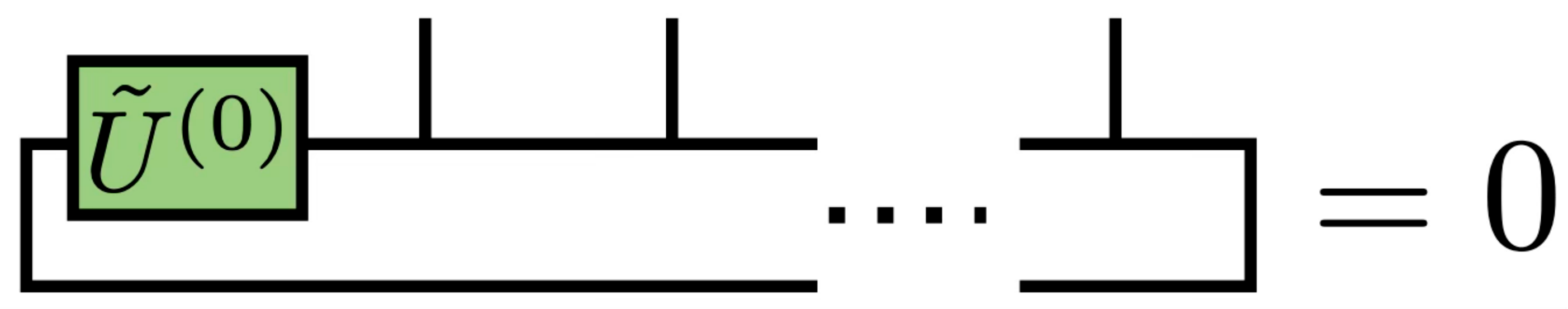}
    \end{aligned}
    \end{equation}
\end{lem}
Thus, all post-measurement states are equivalent to the target state, up to on-site quasi-commuting unitaries, and can therefore be corrected (perhaps by using an additional auxiliary system as discussed above). In summary, we have successfully, deterministically transformed a product state to a state in the phase $(H,\mu)$, and thus, by the arguments above, shown that all phases trivialize for finite abelian groups.

\section{Non-Abelian Symmetries}\label{sec:NonAbelianSymmetries}
We now want to study the phase diagram under non-abelian symmetries. In the case of abelian symmetries, the transformation from any state to the trivial phase was straightforward. One could simply locally measure all physical sites in a common eigenbasis of the symmetry and thereby always deterministically reach a symmetric, pure product state. A direct consequence of this was that SPT phases trivialized. However, for non-abelian symmetries, such an eigenbasis does not need to exist, which is why we have to find a more sophisticated protocol to achieve the desired transformation. First, one notices the following.
\begin{res}\label{res:LOCC}
Let $G$ be a non-abelian group such that $\omega_g^\mu \otimes \omega^{*\mu}_g$ is non-commutative for all nontrivial projective $\mu$-irreps $\omega_g^\mu$~\footnote{As with linear representations, a projective representation is said to be irreducible if it has no nontrivial invariant subspaces. We refer to an irreducible
representation corresponding to the cohomology class $\mu$ as a $\mu$-irrep.}. Then, for any translationally invariant SPT state $\ket{\psi_n}$ associated with $\mu$, there is some $n_0$ such that $\forall n>n_0$ $\ket{\psi_n}$ cannot be deterministically converted via symmetric LOCC to a product state.
\end{res}

The proof is provided in Appendix~\ref{app:proofLOCC}. The idea is to consider the SPT state as a bipartite system, in which case one can see that the non-abelian symmetry ensures that, for any LOCC protocol, at least one outcome will remain entangled in the bipartition. The reason for that is the fact that the state, considered as a bipartite state, is locally supported on the \textit{{non-abelian} subspace}, which refers to the local subspace on which the non-abelian sub-representations in $U_g$ have support.

For instance, the dihedral group with eight elements, $D_8$, is non-abelian and has only one nontrivial projective irrep, $\omega_g$. This irrep has the property that $\omega_g \otimes \omega_g^*$ is non-abelian. Therefore, no nontrivial SPT state under $D_8$ can be deterministically transformed to a product state with symmetric LOCC.

Thus, local measurements with classical communication are generally not powerful enough to deterministically transform any non-abelian SPT state to the product state.
Nevertheless, one might hope to find measurements such that all post-measurement states belong to the trivial phase and can thus be transformed to the product state with a short-depth circuit. However, this too seems unreasonable as one would expect that at least some measurement outcomes yield (potentially long-range) entangled states (see the Result~\ref{res:NonAbelianSPT} below). Thus, instead of strictly deterministic transformations, we will consider asymptotically deterministic transformations in the following. 

\subsection{Join-And-Measure Protocols}
We will consider protocols that can be appropriately described as join-and-measure protocols. The general idea is as follows; first, each site performs a symmetric projective measurement yielding two qualitatively different types of outcome: either the local site is projected onto a one dimensional subspace, carrying a 1D irrep of $G$. The local site is then in a symmetric pure state that is disentangled from the rest of the state, which we will call a \emph{successful outcome}; or, the local site is projected onto the non-abelian subspace. We will call this type of outcome an \emph{unsuccessful outcome}. All local sites that obtain unsuccessful outcomes are in a (potentially long-range) entangled post-measurement state, which we refer to as the \emph{error state}.

To further disentangle the error states, one uses circuits to join neighboring sites of the error state together and then performs a projective symmetric measurement again, this time with respect to $U_g\otimes U_g$  (see Fig.~\ref{fig:non-abelian-Protocol}). The idea is that two copies of a non-abelian representation may contain abelian sub-representations. Therefore, while the error states are strictly-locally supported only on the non-abelian subspaces, they may have abelian support when we consider two sites together. In this case, a subsequent symmetric measurement on pairs of sites will, with some probability, yield again pure product states and a smaller error state. Repeating this procedure may eventually result in a state in the trivial phase.

However, there are several subtleties to this construction. As we are limited to $O(\text{polylog}(n))$-depth circuits, one needs to ensure that the distance between neighboring sites of error states does not become too large when considering the thermodynamic limit. Moreover, we need to ensure the measurements act on $O(1)$ sites. Finally, it is a priori not clear how tensor products of non-abelian irreps decompose and how large the local abelian support of the error states is after pairing sites. Nonetheless, for certain cases, one can show that these protocols terminate and that the probability of obtaining non-correctable error states vanishes in the thermodynamic limit. In these cases, join-and-measure protocols allow for asymptotically deterministic transformations to the trivial phase. In the following, we will delve into this by considering a specific example.

\subsection{Trivialization of SPT Phases under the Non-Abelian Group $D_8$}\label{sec:SPTtrivial}
Consider the non-abelian group $D_8$; the dihedral group of order eight describing the symmetries of a square. It is the smallest non-abelian group with nontrivial SPT phases. It has five linear irreps; four abelian irreps denoted by $\chi_g^{(i)\in \{0,\dots, 3\}}$, and one non-abelian irrep denoted by $U^{(4)}_g$. Later, it will become important, that $U_g^{(4)}\otimes U_g^{(4)}\cong\bigoplus_{i=0}^3 \chi_g^{(i)}$; that is, the tensor square of the non-abelian irrep completely decomposes into abelian irreps. Finally, we choose the following linear representation of $D_8$ 
\begin{equation}
    U_g=\omega_g\otimes\omega_g^*=[\chi^{(0)}_g]_{\Phi^+}\oplus [\chi^{(1)}_g]_{\Phi^-}\oplus [U^{(4)}_g]_{\Psi^\pm}.
\end{equation}
The subscript denotes the corresponding invariant subspaces of the Hilbert space. We can then prove the following result.
\begin{res} \label{res:NonAbelianSPT}
    Any SPT state under $D_8$ can be transformed asymptotically-deterministically to the trivial phase, and vice versa.
\end{res}

To show Result \ref{res:NonAbelianSPT}, consider the state
\begin{equation}
    \ket{{\rm SPT}_n}=\bigotimes_{i=1}^n\ket*{\Phi^+}_{R_i,L_{i+1}},
\end{equation}
where $R_n=L_1$. This state is symmetric with respect to $U_g$ and belongs to the SPT phase associated with $\omega_g$. We now wish to transform this state to the trivial state $\ket{\rm Triv}=\ket*{\Phi^+}^{\otimes n}$ via a join-and-measure protocol.

\begin{figure}[t]
        \includegraphics[width=1\linewidth]{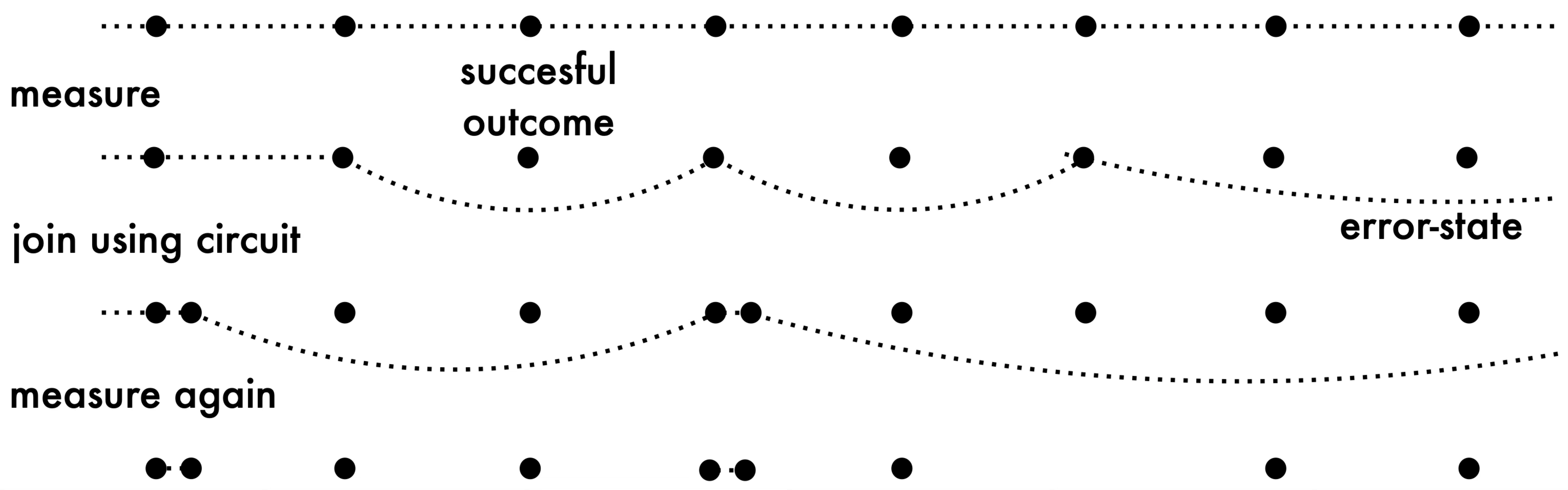}
        \caption{Join-and-measure protocols consist of the following steps: starting from an initial state, a local measurement is performed; a successful outcome results in a pure product state, and an unsuccessful measurement outcome leads to a long-range entangled state, the error-state. In the next step, a SWAP-circuit of a certain depth $l$ is used to bring nearest sites of the error state together. For instance, for $l=2$, the segment of the error-state in the second line is correctable, but the one in the third line is not.}
        \label{fig:non-abelian-Protocol}
\end{figure}

Let us define the following symmetric projective measurement $\{P_k\}_{k\in\{0,1,f\}}$, where
\begin{align}
    P_0 =& \ketbra*{\Phi^+}, \quad P_1 = \ketbra*{\Phi^-}, \label{eq:nonabelianSPTeqsStart}\\
    P_f =& \ketbra*{\Psi^+}+\ketbra*{\Psi^-},
\end{align}
where, clearly, $\qty[U_g,P_k]=0$ for all $k$. Moreover, the following relations hold;
\begin{align}
    \ket*{\Phi^-}&=\id\otimes Z\ket*{\Phi^+}, \\
    \id\otimes Z\ket*{\Phi^+}&=Z\otimes \id\ket*{\Phi^+}, \label{eq:trick}\\
    P_f \left(\id\otimes Z\right)&=-P_f \left(Z\otimes \id\right), \\
    P_f \left(\id\otimes Z \right)&= \left(\id\otimes Z \right) P_f,\\
    U_g \left(\id \otimes Z \right)& \propto \left(\id \otimes Z \right) U_g,\quad  \forall g\in G.  \label{eq:nonabelianSPTeqsEnd}
\end{align}
To begin our analysis, we need to know the probability distribution of the output strings. If each physical site performs the above projective measurement, one obtains an outcome $\vec{k}\in \qty{0,1,f}^{\times n}$. The outcomes $k=0,1$ are successful outcomes as locally the post-measurement state is a disentangled symmetric pure state. From Eqs.~\eqref{eq:nonabelianSPTeqsStart}-\eqref{eq:nonabelianSPTeqsEnd}, it is clear that it does not matter which of the successful outcomes any given site measures. All outcome strings with the same unsuccessful sites are (quasi-commuting) on-site unitary equivalent and therefore occur with the same probability and are equivalently correctable. Thus, it is sufficient to gather successful outcomes and consider the probability of output string $\vec{x}\in\{s,f\}$, where $x_i=s$ corresponds to $k_i\in\{0,1\}$.

Let us now consider the completely successful outcomes, i.e., $\vec{x}=s^{\times n}$ (equivalently, $\vec{k}\in\{0,1\}^{\times n}$). We have
\begin{align}
    p(s^{\times n})&=
     \sum_{\vec{k}\in \{0,1\}^N}\abs*{\bra*{\Phi^{(\pm)^{k_i}}}^{\otimes n}\bigotimes_i \id_{L_i}\otimes (Z^\dagger_{R_i})^{k_i}\ket*{{\rm SPT}}}^2 \\
    &=\frac{1}{4^{n-1}}\sum_{\vec{k}}\abs*{\bra*{\Phi^+}\prod_i\id_{L_1}\otimes (Z^\dagger_{R_1})^{k_i}\ket*{\Phi^+}}^2 \\
    &=\frac{1}{2^{n-1}}.
\end{align}
For the first equality, we consider the cumulative probability of the $2^n$ completely successful outputs. For the second equality, one repeatedly uses Eq.~\eqref{eq:trick}. For the last equality, one uses that the overlap in the sum is one if $\prod_i\id_{L_1}\otimes (Z^\dagger_{R_1})^{k_i}=\id$ and zero otherwise.

In the event of at least one unsuccessful outcome, Eq.\eqref{eq:nonabelianSPTeqsStart}-\eqref{eq:nonabelianSPTeqsEnd} ensure the probability $p(\vec{x})$ can at most depend on the number of unsuccessful outcomes but not on their location. A similar calculation shows that the probability of obtaining $\abs{f}$ unsuccessful outcomes is given by
\begin{equation}
    p(\vec{x})=2^{n-\abs{f}}\frac{1}{4^{n-\abs{f}}}\bra*{{\rm SPT}_\abs{f}} P_f^{\otimes \abs{f}} \ket*{{\rm SPT}_\abs{f}}.
\end{equation}
By evaluating the expectation value above, one can show that it evaluates to zero whenever $\abs{f}$ is odd and $1/2^{\abs{f}-1}$ otherwise. The proof can be found in Appendix~\ref{app:evenLemma}. It uses the fact that relations similar to the following hold;
\begin{equation}\label{eq:P3network}
\begin{aligned}
    \includegraphics[width=0.48\linewidth]{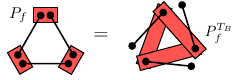}
\end{aligned}.
\end{equation}
Here, $T_B$ denotes the partial transpose on the second system. To summarize, one finds that
\begin{align}\label{eq:probDistr1}
    p(\vec{x})=\begin{cases}
        \frac{1}{2^{n-1}} &\abs{f}\text{  even},\\
        0 & \abs{f}\text{  odd}.
    \end{cases}
\end{align}

Now equipped with the probability distribution on the outputs, we consider which outcomes with at least one unsuccessful outcome can be corrected. Any post-measurement states can always be brought to the form
\begin{equation}
    \ket*{\psi(\vec{x})}\propto \ket*{\Phi^+}^{\otimes n-\abs{f}}\otimes P_f^{\otimes \abs{f}} \ket*{{\rm SPT}_\abs{f}}.
\end{equation}
To transform the error-state, $\ket{\rm err}=P_f^{\otimes \abs{f}} \ket*{{\rm SPT}_\abs{f}}$, to a product state, one notices that it can be written as
\begin{equation}\label{eq:errorstate}
    \ket{\rm err}\propto \ket{0110}^{\abs{f}/2}+\ket{1001}^{\abs{f}/2},
\end{equation}
which is a GHZ-type state, and where we have paired up neighboring sites. Physically, this may be achieved by applying a quantum circuit, with a depth depending on the maximum distance between sites of the error-state. In order to break the GHZ entanglement, one notices that there exist states $\eta_0\propto \ket{0110}+\ket{1001}$, and $\eta_1\propto \ket{0110}-\ket{1001}$ that are symmetric with respect to $U_g^{\otimes 2}$, and $\ket{\rm err}\propto \ket{\eta_0+\eta_1}^{\abs{f}/2}+\ket{\eta_0-\eta_1}^{\abs{f}/2}$. Thus, we can measure the error state on pairs of sites in the basis $\{\ketbra{\eta_0},\ketbra{\eta_1},P_{\eta_1\eta_2}^\perp\}$ to obtain a symmetric pure product state, which can be locally corrected to the state $\ket{\Phi^+}$. Thus, for this example, the protocol will terminate after two rounds, using measurements on $O(1)$ systems. Note, this is a consequence of the fact that the non-abelian linear irreps of $D_8$ become abelian under tensor powers. 

Finally, we need to ensure that the circuits we use to pair up the sites have $O(\text{polylog}(n))$-depth. That is, it may happen that the distances between sites of the error state become too large to be correctable with short-depth circuits. However, in Appendix~\ref{app:probError}, we show that, in the limit $n\rightarrow\infty$, the probability of obtaining error-states that require more than log$(n)$-depth circuits to join nearest neighbors tends to zero. We can thus transform any SPT state under $D_8$ asymptotically-deterministically to the trivial phase.

Note, that this also implies the reverse transformation is possible as one can prepare two copies of the SPT state (which belongs to the trivial phase and can therefore be reached via short-depth circuits from a product state) and then transforming one copy to a product state.

Finally, we wish to emphasize that the protocol above relied on the fact that tensor-squares of all non-abelian irreps of $D_8$ are completely decomposable into abelian irreps. This feature is not unique to the group $D_8$.

\subsection{Transforming GHZ States to the Trivial State under $D_8$}\label{sec:GHZproduct}
Finally, let us consider non-normal states and show that there is an asymptotically deterministic transformation to the trivial phase. Specifically, we show that a similar protocol as above can be used to transform the GHZ state to the product state under the symmetry $D_8$. As a representation, we choose the regular representation of $D_8$, i.e., the eight dimensional representation given by
\begin{multline}
    U_g=[\chi^{(0)}]_{\varphi_0}\oplus [\chi^{(1)}]_{\varphi_1}\oplus [\chi^{(2)}]_{\varphi_2}\oplus [\chi^{(3)}]_{\varphi_3}\\ \oplus [U^{(4)} \oplus U^{(4)}]_{P_f},
\end{multline}
where $P_f$ denotes the projector onto the non-abelian subspace and
\begin{equation}
    \ket*{\varphi_i}=\frac{1}{2\sqrt{2}}\sum_{g\in D_8} e^{-i\varphi_g^i}\ket{g}=\tilde{Z}^i\ket*{\varphi_0},
\end{equation}
with $\ket*{\varphi_0}=1/\sqrt{|G|} \sum_{g\in G} \ket{g}$, $\{e^{-i\varphi_g^i}\}$ are 1D irreps of $D_8$, and
\begin{equation}
    \tilde{Z}^i={\rm diag}(\{e^{-i\varphi_g^i}\}).
\end{equation}
As an initial state, we consider the GHZ state on $n$ particles
\begin{align}
    \ket*{{\rm GHZ}_n^8} =& \frac{1}{2\sqrt{2}}(\ket{00\dots 0}+\dots +\ket{77\dots 7}).
\end{align}
It is easily verified that this state belongs to the symmetry breaking phase $(H,\mu)=(1,0)$ (see Appendix~\ref{app:noninjectivesym}). 

A local symmetric projective measurement is then given by $\{P_k\}_{k\in\{0,\dots,3,f\}}=\{\tilde{Z}^i\ketbra*{\varphi_0}(\tilde{Z}^\dagger)^i\}_{i=0,\dots, 3}\cup \{P_f\}$. The first four outcomes are successful outcomes and the fifth outcome is the only unsuccessful outcome. 
For the state and the measurement defined above, the following relations hold;
\begin{align}
    (\tilde{Z}^i_k)^\dagger \ket*{{\rm GHZ}_n^8} &=(\tilde{Z}^i_j)^\dagger\ket*{{\rm GHZ}_{n}^8},\quad \forall\, j,k\in [n], \label{eq:relation1} \\
    \braket*{\varphi_0}{{\rm GHZ}_{n}^8}&=\frac{1}{2\sqrt{2}}\ket*{{\rm GHZ}_{n-1}^8}, \label{eq:relation2} \\
    [P_{f},(\tilde{Z}^i)^\dagger]&=0. \label{eq:relation3}
\end{align}
As before, these relations ensure that it is sufficient to coarse-grain the outcomes and only consider whether outcomes are successful or unsuccessful. That is, consider the probability distribution for $\vec{x}\in\{s,f\}^{\times n}$.

From the relations above one can deduce that, after the first round of measurements, the probability distribution of outcomes is again flat on even strings; namely, one finds
\begin{align}\label{eq:probDistr}
    p(\vec{x})=\begin{cases}
        \frac{1}{2^{n-1}} &\abs{f}\text{  even},\\
        0 & \abs{f}\text{  odd},
    \end{cases}
\end{align}
(see Appendix~\ref{app:GHZ} for a proof). Moreover, one can show that all post-measurement states are, up to quasi-commuting local unitary corrections, of the form
\begin{equation}
    \ket{\psi(\vec{x})} \propto \ket{\varphi_0}^{\otimes n-\abs{f}} \otimes P_f^{\otimes \abs{f}}\ket*{{\rm GHZ}_{\abs{f}}^8}.
\end{equation}
To transform the error-state to a product state, one again uses the fact that the subspace spanned by $P_f\otimes P_f$ can be completely decomposed into 1D irreps. Therefore, after bringing neighboring sites together by using a log$(n)$-depth circuit, one can perform a measurement that with unit probability disentangles the error state. Once again, one needs to ensure that the probability of obtaining non-correctable error configurations vanishes in the thermodynamic limit; however, the same reasoning as in the case of SPT states also applies here.

\section{Discussion and Conclusions}
We have studied how symmetry-preserving measurements with feedforward alter the phase classification of MPS in the presence of global on-site symmetries. Let us stress here again that this classification in natural in the context of quantum simulation. Moreover, this work is a natural extension of the results in Ref.~\cite{Piroli2021, malz2023preparation} to include symmetry constraints. 

We completely solve this classification for finite abelian groups, showing that the phase diagram trivializes. Therefore, in contrast to Ref.~\cite{Schuch2011_Phases2,Chen2011_Phases1}, adding symmetries does not alter the phase classification of Ref.~\cite{Piroli2021}. In particular, for any pair of phases, we construct a protocol that uses finite-depth symmetric circuits, a constant number of auxiliary systems, and two rounds of symmetric measurements to transform any state in the initial phase to a state in the target phase. The non-abelian case is more challenging as one cannot deterministically transform any state on-site to a product state. Nonetheless, by providing explicit examples, we have shown that with $\log(n)$-depth circuits and two round of symmetric measurements, non-abelian SPT phases can trivialize. 

Looking forward, it would be interesting to investigate to what extent join-and-measure protocols can be used to asymptotically-deterministically transform between nontrivial phases of groups that have a more intricate structure than the examples we have considered. Moreover, it would be interesting to ascertain lower-bounds on the cumulative circuit depth required to asymptotically trivialize non-abelian phases. Finally, the phase classification of matter with symmetries has been extended to higher dimensions (see, e.g., Ref.~\cite{Chen2013_HigherDimSPT}), and so it would be interesting to also extend this work to higher dimensions.

\ \\

\section{Acknowledgments}
We thank J. Ignacio Cirac, Norbert Schuch, and Alex Turzillo for fruitful discussions. DG, TK and BK acknowledge financial support from the Austrian Science Fund (FWF) through the grants SFB BeyondC (Grant No. F7107-N38), and P 32273-N27 (Stand-Alone Project). Furthermore, we acknowledge the BMW endowment fund. This publication has received funding under Horizon Europe programme HORIZON-CL4-2022-QUANTUM-02-SGA via the project 101113690 (PASQuanS2.1). GS is supported by the Alexander von Humboldt Foundation.

\bibliography{bibliography}

\appendix
\section{Symmetries and Phases of MPS}\label{app:noninjectivesym}
In this Appendix, we review details on how tensors of MPS transform under global on-site symmetries and how this then characterizes the phases of matter under symmetries (without measurements). We include this appendix for clarity and so that this work is self-contained. However, these details can be found in Refs~\cite{Cirac2017_FundThm1, Schuch2011_Phases2,PerezGarciaEtAl2007_MPSrepresentations}. In Section~\ref{app:AppendixAInjectiveMPS}, we give a mathematical background to projective representations. In Section~\ref{app:AppendixANonInjectiveMPS}, we explain in more detail the action of symmetries on non-injective MPS. In Section~\ref{sec:BlockingAndLocalSymmetry}, we discuss the relationship between the original on-site symmetry and the effective physical symmetry. In Section~\ref{sec:AppendixAPhasesHamiltonians}, we discuss the classification of phases in Ref.~\cite{Schuch2011_Phases2} via parent Hamiltonians. Finally, in Section~\ref{sec:AppendixAPhasesCircuits}, we provide more detail on how the phases in the Hamiltonian picture relate to phases in the circuit picture.

\subsection{Projective Representations} \label{app:AppendixAInjectiveMPS}

In this section, we provide the relevant background for projective representations, appearing, for instance, in the action of global on-site symmetries of MPS in Eqs~\eqref{eq:injectivesym} and \eqref{eq:noninjectivesym} in the main text.

To begin, projective representations differ from linear representations in that they must only be closed up to a phase, i.e.,
 \begin{equation}
     \omega_g \omega_h = \gamma(g,h) \omega_{gh},
 \end{equation}
 for all $g, h\in G$, where $\gamma: G\times G \rightarrow U(1)$ is referred to as a cocycle. Two projective representations are said to be \emph{projectively equivalent} if there is a set of phases $\{\nu(g)\}_{g\in G}$ and a unitary $V$ such that
\begin{equation}
    \tilde{\omega}_g = \nu(g) V^\dagger\omega_g V.
    \label{eq:ProjectiveEquivalence}
\end{equation}
This induces an equivalence relation on the cocycles; namely 
\begin{equation}
    \gamma (g,h)\sim \frac{\nu(gh)}{\nu(g) \nu(h)} \gamma (g,h)
\end{equation} for any $\nu: G\rightarrow U(1)$. The induced equivalence classes of the $\gamma (g,h)$ are isomorphic to the second cohomology group of $G$ over $U(1)$, $H^2(G,U(1))$, where the group operation in $H^2(G,U(1))$ is induced by tensor products of the projective representation, i.e.,
\begin{equation}
    \omega_g^{\mu_1}\otimes \tilde{\omega}_g^{\mu_2}\cong \hat{\omega}_g^{\mu_1 \oplus \mu_2}.
    \label{eq:TensorProductProjReps}
\end{equation}
Thus the $\mu$ are typically referred to as cohomology classes. Correspondingly, we refer to a projective representation with a cocycle belonging to $\mu$ as a $\mu$-projective representation, and write it as $\omega_g^\mu$.

As with linear representations, a projective representation is said to be \emph{irreducible} if it has no nontrivial invariant subspaces. We refer to an irreducible $\mu$-projective representation as a $\mu$-irrep. Moreover, any $\mu$-projective representation of a finite group is completely reducible to a direct sum of $\mu$-irreps \cite{karpilovsky1994_ProjTraceCondition}. For finite groups, there is a finite number of $\mu$-irreps (up to projective equivalence) for each cohomology class. Moreover, for abelian groups, there is a unique $\mu$-irrep per class \cite{Backhouse1972}.

Finally, we note that the cohomology class associated with an MPS and a symmetry is not stable under tensor products. Namely, consider an injective MPS, $\ket{\psi_n[A]}$, such that the corresponding action of $U_g$ in the virtual space is given by $\omega_g^\mu$. Then, taking two copies, it is easy to see that the corresponding action of $U_g^{\otimes 2}$ in the virtual space of $\ket{\psi_n[A]}^{\otimes 2}$ is $(\omega_g^{\mu})^{\otimes 2}$ which, by Eq.~\eqref{eq:TensorProductProjReps}, belongs to the cohomology class $\mu\oplus\mu\in  H^2(G,U(1))$. In particular, as $H^2(G,U(1))$ is a group, there is an inverse element, $\mu^{-1}\in H^2(G,U(1))$, such that if $\tilde{U}_g$ acting on $\ket{\psi_n[\tilde{A}]}$ induces $\omega_g^{\mu^{-1}}$ in the virtual space, then the action of $U_g\otimes \tilde{U}_g$ on $\ket{\psi_n[A]}\otimes\ket{\psi_n[\tilde{A}]}$ corresponds to the trivial cohomology class (a linear representation).

\subsection{Transformation of Non-Normal MPS under Global On-Site Symmetries}
\label{app:AppendixANonInjectiveMPS}
In this section, we discuss the action of symmetries on non-injective MPS, as derived in Ref.~\cite{Schuch2011_Phases2}, see also Eq.~\eqref{eq:noninjectivesym} in the main text. 
More precisely, we will discuss how the fact that $U_g$ is a linear representation gives $P_g$, $e^{i\phi_g^\alpha}$ and $\omega_{g,\alpha}$ their further structure.

To begin, $P_g$ forms a permutation representation of the group $G$ by permuting the blocks of $A$. A permutation representation $\pi : G \rightarrow \text{Sym}([m])$ is a homomorphism from the group $G$ to the symmetric group on the set $[m]$. This allows one to write
\begin{align}
    P_g &= \sum_{\alpha\in[m]} \ketbra*{\pi_g(\alpha)}{\alpha},
    \label{eq::PermutationAction}
\end{align}
where $m$ is the number of blocks in the canonical form. Now consider the orbits of each block $\alpha\in [m]$ under the permutation action, $G~\cdot~\alpha  = \{ \pi_g(\alpha) : g\in G\}$. It may be the case that some nontrivial subset of $[m]$ is left invariant by $P_g$; that is, it could be that there is an $S \subsetneq [m]$ such that $G \cdot S = S$. In these cases, this block-structure is not "protected" by the symmetry and is referred to as "accidental degeneracy". In this paper, we only consider MPS with symmetries such that there are no such nontrivial invariant subsets.

So from now on let us consider the case in which $[m]$ has no nontrivial subsets under the action of $G$, i.e., $G~\cdot~\alpha =[m]$ for all $\alpha\in[m]$. 
Let us consider the orbit of $0\in [m]$ (i.e., the first block in $A$). Whilst $0$ cannot be invariant under all $g\in G$ (unless $[m]=\{0\}$), it may be invariant under a subgroup. That is, we can consider the stabilizer subgroup of $G$ with respect to $0$, $H_{0}$, given by
\begin{equation}
    H_{0} =\{ g\in G : \pi_g(0) = 0\} \le G.
\end{equation}
The orbit stabilizer theorem states that
\begin{equation}
    |G~\cdot~0| = \frac{|G|}{|H_{0}|}.
\end{equation}
Note, as $ G~\cdot~0 = [m]$, we therefore have that the number of blocks in the tensor is given by $|G|/|H_{0}|$.

Next, one can consider the left cosets induced by this subgroup
\begin{equation}
    [G:H_{0}] = \{k_0 H_{0},k_1 H_{0},\dots, k_{|G/H_0|-1} H_{0}\}.
\end{equation}
By construction $P_{g_1}=P_{g_2}$ if and only if $g_1,g_2$ belong to the same coset. Moreover, as $G\cdot 0 = [m]$, for each coset one may choose a representative, $k_\alpha \in G$, such that $k_\alpha$ maps the block $0$ to the block $\alpha$, i.e.,
\begin{equation}
    \pi_{k_\alpha}(0) = \alpha.
    \label{eq:AppendixDefOfK}
\end{equation}
Having chosen the representations, one can uniquely define maps $h:G\times [m] \rightarrow H_0$ and a maps $\xi:G\times[m]\rightarrow [m]$ by the relation
\begin{equation}
    g k_\alpha = k_{\xi(g,\alpha)} h(g,\alpha).
    \label{eq:AppendixDefofHandXi}
\end{equation}
That is, for all $g\in G $ and $ \alpha\in [m]$, the group element $g k_\alpha \in G$ [with $k_\alpha$ defined by Eq.~\eqref{eq:AppendixDefOfK}] belongs to the $\xi(g,\alpha)$ coset.

Observe, that Eqs~\eqref{eq:AppendixDefOfK}~and~\eqref{eq:AppendixDefofHandXi} in-fact completely specify the entire permutation action as
\begin{align}
    \pi_g(\alpha) &= \pi_{g} [\pi_{k_\alpha}(0)]\\
    &= \pi_{g k_\alpha} (0)\\
    &= \pi_{k_{\xi(g,\alpha)} h(g,\alpha) } (0)\\
    &= \pi_{k_{\xi(g,\alpha)}} (0)\\
    &= \xi(g,\alpha).
\end{align}
Thus, one could have equivalently characterized the permutation action by specifying $H_0$. However, the choice of $0$ was completely arbitrary. In fact, as all $a\in[m]$ yield the same orbit (as there is only one invariant subspace), the stabilizers $H_{\alpha\in[m]}$ are all isomorphic. Thus, the permutation action of $U_g$ is, up to relabeling, completely specified by a subgroup $H\le G$.

Before moving on, we make two final comments regarding the permutation representation. Firstly, the subgroup, $H$, need not be normal. However, if it is (as, for example, in the case of abelian groups), then $[G:H]$ can be given a well-defined group multiplication, yielding the quotient group $G/H$, with a homomorphism $k:G\rightarrow G/H$. In such cases, $P_g=P_{k(g)}$ is a representation of $G/H$. This will become important later in the construction of protocols showing the trivialization of the phases of non-injective MPS under abelian symmetries, cf. Appendices~\ref{app:Example} and~\ref{app:abelian}. Secondly, as we have assumed that there is only one invariant subspace under the action of the permutation group, if $H=G$, then $[m]=\{0\}$, i.e., our tensor has only one block. Thus, the MPS is normal.

\subsection{Comment on Blocking and Locality of the Physical Symmetry} \label{sec:BlockingAndLocalSymmetry}
In this section, we discuss the mathematical details of how the effective symmetry on the blocked MPS in canonical form relates to the global on-site symmetry on the MPS before blocking to canonical form. We chose to define our operations with respect to some pre-determined global on-site symmetry as this fits more naturally with the circuit picture of phases. However, in the previous section, the symmetry action was defined on the canonical form. This canonical form is defined after blocking. This means that the symmetry considered in, for example, Eq.~\eqref{eq:noninjectivesym} in the main text is not necessarily the original symmetry of the system. Here we explicitly provide the mathematical background regarding how this tension is resolved.

To make this concrete, consider an MPS $\ket{\psi_n[A]}\in\mathcal{H}_d^{\otimes n}$ with a global on-site symmetry $U_g$. Consider this to be the original system. Let us also say we must block this tensor at least $l$ times in order to bring the tensor to canonical form (that is, $l$ is the so-called \textit{block-injectivity length} of $A$). We write $(A^{(l)})^{i_1,\dots,i_m}=\prod_{j=1}^m A^{i_j}$. The canonical form, $\tilde{A}^i=\oplus_\alpha \tilde{A}_\alpha^i$, then has the property that
\begin{equation}
    \ket*{\psi_{m}[\tilde{A}]} =|\psi_{m}[A^{(l)}]\rangle = \ket*{\psi_{n=ml}[A]} \in \mathds{C}_d^{\otimes n}.
\end{equation}
Note, we have deliberately expressed this state as an element of the original Hilbert space. Consequently, it is only defined for $n$ equal to integer multiples of $l$. However, as we are ultimately interested in the thermodynamic limit, this is not a problem.

Now, if we consider the reduced density matrix of this state on $l$ sites, we find that the rank is equal to  $\sum_\alpha D_\alpha^2$, where $D_\alpha$ is the dimension of the blocks $A_\alpha^i$~\cite{PerezGarciaEtAl2007_MPSrepresentations}~\footnote{Note that $D=\sum_\alpha D_\alpha$. Thus $\sum_\alpha D_\alpha^2 \le D^2$.}. In particular, it may be less than $d^l$. In this case, one can verify that the structure of Eq.~\eqref{eq:noninjectivesym} in the main text, with $U_g := U_g^{\otimes l}$, means that $U_g^{\otimes l}$ (in an appropriate physical basis) can be written as
\begin{equation}
    U_g^{\otimes l} \cong_{LU} \begin{pmatrix}
  W_g  & 0 \\
 0  & U_g^{\rm extra} 
 \end{pmatrix},
\label{eq:UgBlockStructure}
\end{equation}
where  $U_g^{\rm extra}$ acts locally on the space where $\ket{\psi_{m}[A]} \in \mathds{C}_{d^l}^{\otimes m}$ has no support, and
\begin{equation}
    \label{eq:AppSymmetryPhysicalSubspace}
    W_g = \left[\bigoplus_{\alpha\in K}  e^{i\phi_g^\alpha} (\omega^\mu)_{h(g,\alpha)}^* \otimes (\omega^\mu)_{h(g,\alpha)}\right] \left(P_g \otimes \id^{\otimes 2}\right).
\end{equation}
That is, the nontrivial action of the symmetry on $\ket{\psi_n[A]}$ on $l$ sites is equivalent to the action of the symmetry in the virtual space (after blocking and bringing it to CF).

Furthermore, it is clear from Eq.~\eqref{eq:noninjectivesym} in the main text, that blocking past the block-injectivity length will at most change the phases $e^{i\phi_g^\alpha}$. Indeed, in the case of finite groups, it is clear that after blocking a finite number of times these phases can be eliminated. Thus, for any $l$ greater than or equal to the  block-injectivity length, Eq.~\eqref{eq:UgBlockStructure} holds.

\subsection{Phases of Matter under Symmetries using Parent Hamiltonians} \label{sec:AppendixAPhasesHamiltonians}

We are now in a position to review the classification of phases of MPS.

\begin{thm}  \label{thm:phasesnoninjective}\cite{Chen2011_Phases1, Schuch2011_Phases2}
    Let $G$ be a finite group and $U_g$ be a global on-site symmetry. Two symmetry-protected MPS belong to the same phase under $G$ if and only if the action of the symmetry on the tensor in canonical form corresponds to the same subgroup, $H\le G$, and the same cohomology class, $\mu\in H^2(H,U(1))$. \label{thm:NonInjectivePhases}
\end{thm}

An example of this classification is provided in Appendix~\ref{app:Example}.

Some comments are in order. Firstly, the fact that the phase depends only on $(H,\mu)$, both of which correspond to group properties, reveals a small technicality regarding the phase diagram associated with a given $U_g$; it may be the case that not every global on-site symmetry is capable of supporting every phase of matter associated with $G$. That is, for some $U_g$, there may be certain phases for which no MPS exists that belongs to this phase. In order for a global on-site symmetry, $U_g$, to be able to support every phase, $(H,\mu)$, of $G$, there must be an $m\in\mathds{N}$ such that $U_g^{\otimes m}\cong W_g(H,\mu) \oplus U_g^{\rm extra}$ [with $W_g(H,\mu)$ as in Eq.~\eqref{eq:AppSymmetryPhysicalSubspace}]. This may not be the case, even when $U_g$ is a faithful representation of a finite group. As a concrete example, consider the finite, abelian group, $K_4\cong \mathds{Z}_2\times \mathds{Z}_2$, and the representation $U_{(i,j)\in \mathds{Z}_2\times \mathds{Z}_2} = (-1)^i \oplus (-1)^j$. It is easy to verify that this representation is indeed faithful. Now consider the phase $(H=K_4,\mu=1)$. One can verify that the associated symmetry action is given by $W_{(i,j)}=1\oplus (-1)^i \oplus (-1)^j\oplus(-1)^{ij}$. However, $W_{(i,j)}$ is not found as a sub-representation of $U_g^{\otimes m}$ for any $m\in\mathds{N}$.

In general, it requires some computation to say which phases of $G$ are supported by a given global on-site symmetry, $U_g$. Nonetheless, this is not a problem for our results as, for example in Result~\ref{res:AbelianPhaseDiagram}, we show that whatever the phase diagram is for a given $U_g$, it trivialises through the inclusion of symmetric measurements and feedforward. This follows as whenever a phase is non-empty, the analysis from the previous section guarantees that there is a $m\in\mathds{N}$ such that $U_g^{\otimes m}$ contains our construction in Eq.~\eqref{eq:Ug} as a sub-representation.

Secondly, we say a state is in the trivial phase if it is in the same phase as a product state. All symmetric product states are in the same phase and correspond to normal, symmetric MPS. Normal symmetric MPS that are not in the trivial phase are referred to as symmetry protected topological (SPT) phases. With the above notation, these correspond to phases labeled by $(H=G, \mu\ne 0)$ and which depend only on the cohomology class of the corresponding projective representation. However, as discussed in Appendix~\ref{app:AppendixAInjectiveMPS}, the cohomology class of projective representations is not stable under tensor products. Consequently, SPT phases are also not stable under tensor products. In particular, for any SPT, $\ket{\psi_n[A(\mu)]}$, there is another state $\ket{\psi_n[A(\mu^{-1})]}$ such that $\ket{\psi_n[A(\mu)]}\otimes \ket{\psi_n[A(\mu^{-1})]} \in (\mathcal{H}^{\otimes 2})^{\otimes n}$ is in the trivial phase.

Thirdly, phases of matter with label $(H\lneq G, \mu)$ correspond to non-normal symmetric MPS and are referred to as symmetry breaking phases. As mentioned before, the parent Hamiltonians of non-injective MPS have degenerate ground state subspaces. In the absence of symmetries, this degeneracy is not stable to perturbations, and thus it is often said that there is no topological order in 1D~\cite{Chen2011_Phases1} (see also discussion in Ref.~\cite{BravyiEtAl2010_StabilityofDegenerateGroundStates, ZengWen2015_GappedQuantumLiquidsAndSLCircuits}). However, in the presence of symmetries, this ground state degeneracy is believed to be stable to symmetric perturbations (see discussion in Ref.~\cite{Schuch2011_Phases2}). In this paper, we make no claims about the stability of phases as classified by Theorem~\ref{thm:NonInjectivePhases}.

\subsection{Correspondence between the Classification of Phases for Hamiltonians and Circuits} \label{sec:AppendixAPhasesCircuits}
The classification from Ref.~\cite{Schuch2011_Phases2} is derived from considering continuous paths of gapped, parent Hamiltonians. It is widely believed that the classification is the same with respect to local symmetric circuits of $O(1)$ gate-size and $O(\polylog n)$ depth, as defined in the main text \cite{Huang2015_symandlindepthcircuits, CoserandGarcia2019_PhasesMixedStatesDissipativeEvo}. We now make some remarks regarding this correspondence.

Note that both definitions invoke the thermodynamic limit. The Hamiltonian classification uses the notion of a spectral gap, while the circuit version is formulated over sequences of states. Moreover, parent Hamiltonians might have a degenerate ground state subspace. In that case, the corresponding state transformed by the quantum circuit is taken to be the (non-injective) MPS formed by the direct sum of the normal tensors spanning the ground state subspace (e.g., a GHZ state, see Ref.~\cite{Schuch2011_Phases2}).

The origin of the polylogarithmic scaling in the circuit can be traced back to the tails occurring when a local Hamiltonian is evolved for finite time. Such unitaries arise in the quasi-adiabatic continuation~\cite{hastings2005quasi} (see the discussion in Ref.~\cite{CoserandGarcia2019_PhasesMixedStatesDissipativeEvo}). Note that a strictly constant depth is not sufficient to reproduce every short-range correlated state in the trivial phase. In particular, no local quantum circuit with depth $l = O(\log N)$ can faithfully approximate a translation-invariant injective MPS with a nonzero correlation length in the thermodynamic limit~\cite{malz2023preparation}.

In the absence of symmetries, the equivalence of phases of MPS can be made exact by constraining the quantum circuit to $O(\polylog n)$ depth with $O(1)$ size gates~\cite{CoserandGarcia2019_PhasesMixedStatesDissipativeEvo, haah2021quantum, Schuch2011_Phases2, malz2023preparation, Huang2015_symandlindepthcircuits}. Moreover, as all of our constructions in this paper only use finitely-many, finite-sized operations, our results hold for any other reasonable classification of phases with ``shallow'' symmetric circuits. In particular, one could also consider finite-depth circuits with polylogarithmic-size gates, for which it is known that state transformations within the same phase are possible even in the presence of symmetries~\cite{CoserandGarcia2019_PhasesMixedStatesDissipativeEvo}.

\section{An Example of Phases under an Abelian Symmetry}
\label{app:Example}
In this appendix, we provide an example to accompany Result~\ref{res:AbelianPhaseDiagram}. To this end, we will consider the finite abelian group $\mathds{Z}_4\times \mathds{Z}_2$, and discuss the associated symmetry protected and symmetry breaking phases of matter. Moreover, we will introduce some notation that will be quite useful in the proof of Result~\ref{res:AbelianPhaseDiagram} in Appendix~\ref{app:abelian}.

\subsection{Phases of Matter under $\mathds{Z}_4\times \mathds{Z}_2$}
The finite abelian group $G= \mathds{Z}_4\times \mathds{Z}_2$ consists of the elements
\begin{equation}
\{ (0,0), (0,1), (1,0), (1,1), (2,0), (2,1), (3,0), (3,1) \}    
\end{equation}
with the group operation given by 
\begin{equation}
    (i,j)\oplus_G (k,l) = (i\oplus_4 k, j \oplus_2 l)
    \label{eq:ExampleGroupAddition}
\end{equation}
where $\oplus_4$ and $\oplus_2$ indicate modulo 4 and 2 addition respectively. In the following, we will refer to the group elements either by writing $g\in G$, or more explicitly by writing $(i,j)\in G$, whatever is more convenient.

As a linear unitary representation of $G$, let us consider $U_{g=(j,k)}=1\oplus (i)^j \oplus (-1)^k$ \footnote{It is easily verified that this is a faithful representation of $G$ that generates all representations under tensor powers, and therefore $U_g$ supports all phases of $G$ (see discussion in Appendix~\ref{sec:BlockingAndLocalSymmetry}). Indeed, it can be easily verified that any smaller representation does not have this property. Note, there are, however, smaller faithful representations.}. Correspondingly, as a local system, let us consider qutrits, $\mathcal{H}=\mathds{C}_3$. As explained in Appendix~\ref{app:noninjectivesym}, we may read off the associated phases of matter solely from the group properties of $G$ and its subgroups. $G$ has six subgroups:
\begin{enumerate}
    \item $\mathds{Z}_4\times \mathds{Z}_2 = G  $,
    \item $2\mathds{Z}_2\times \mathds{Z}_2=\{ (0,0), (0,1), (2,0), (2,1)  \}$\footnote{Note, we use the notation $2\mathds{Z}_2$ to indicate multiplying the elements of $\mathds{Z}_2$ by 2. This is because all these subgroups are groups with respect to the group operation defined in Eq.~\eqref{eq:ExampleGroupAddition}.},
    \item $ \id \times \mathds{Z}_2 = \{ (0,0), (0,1) \}$,
    \item $ \mathds{Z}_4 \times \id = \{ (0,0), (1,0), (2,0), (3,0) \}$,
    \item $ 2 \mathds{Z}_2 \times \id = \{ (0,0), (2,0) \}$,
    \item $ \id \times \id = \{(0,0)\}$.
\end{enumerate}
Of these subgroups, only the first two have a nontrivial cohomology group, both of which are isomorphic to $\mathds{Z}_2$ \cite{subwikiSchurMultiplier}. Therefore, there are eight phases associated with $\mathds{Z}_4\times \mathds{Z}_2$, as shown in Figure~\ref{fig:phasesunderZ4xZ2}.
\begin{figure}[t]
        \includegraphics[width=.35\linewidth]{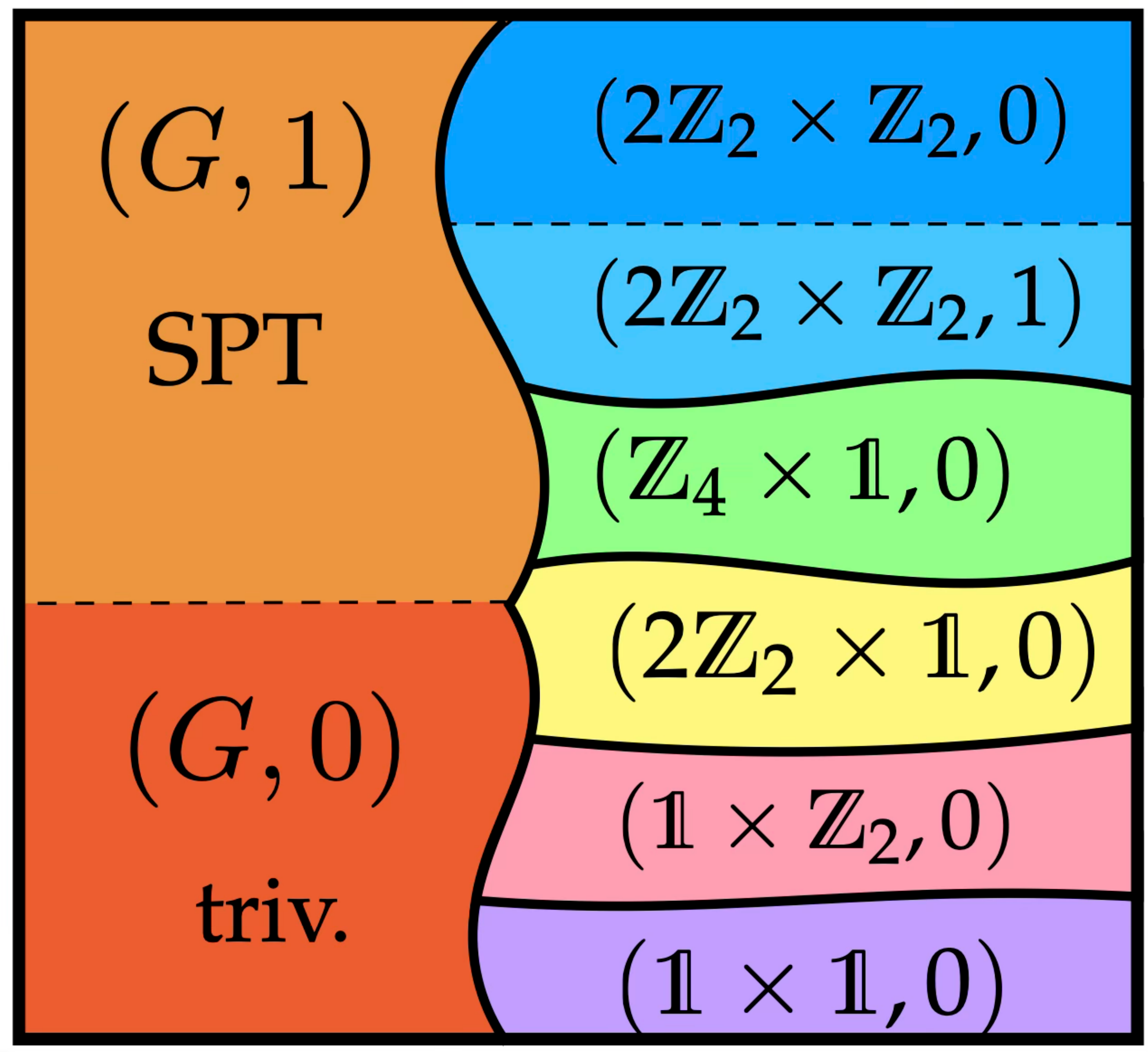}
        \caption{Phases under $\mathds{Z}_4\times \mathds{Z}_2$: The phase $\tilde{H}=\mathds{Z}_4\times \mathds{Z}_2$ with the trivial cohomology class ($\mu=0$) corresponds to the trivial phase; the phase $\tilde{H}=\mathds{Z}_4\times \mathds{Z}_2$ with the nontrivial cohomology class ($\mu=1$) corresponds to the only SPT phase. All other phases are symmetry breaking, and therefore are associated with non-normal MPS.}
        \label{fig:phasesunderZ4xZ2}
\end{figure}

We will now pick a nontrivial phase, $(\tilde{H},\mu)$ \footnote{Note, we use $\tilde{H}$ instead of $H$ to denote the subgroup so to have a more convenient notation later}, and provide the details of the symmetries and the measurements that allow one to transform a product state to this phase. In particular, we choose the subgroup
\begin{align}
    \tilde{H}= 2\mathds{Z}_2\times \mathds{Z}_2
\end{align}
as it contains all the details one needs to fully generalize the protocol. This subgroup is isomorphic to the Klein 4 group, and it is well known that the second cohomology group of the Klein 4 group is isomorphic to $\mathds{Z}_2$, i.e., $H^2[\tilde{H},U(1)]\cong \mathds{Z}_2$. Moreover, we choose the nontrivial cohomology class $\mu=1$. To summarize, we will demonstrate that under $\mathds{Z}_4\times \mathds{Z}_2$, the phase
\begin{equation}
   (\tilde{H},\mu)=(2\mathds{Z}_2\times \mathds{Z}_2, 1)
\end{equation} 
is reachable from a product state.

As explained above, this example will serve as a template from which we can completely generalize our construction to reach not only all phases under $G$, but all phases under any finite abelian group. To this end, we begin with the group properties of $G$ and $H$, including their linear and projective representations, using notation that will allow us to generalize to the full proof in Appendix~\ref{app:abelian}. We will then explicitly define $U_g$ in Eq.~\eqref{eq:Ug} and $\{P_{q,r}\}_{q,r\in G}$ in Eq.~\eqref{eq:Projectors}.

\subsection{Group Properties of $\mathds{Z}_4\times \mathds{Z}_2$ and $\mathds{Z}_2\times \mathds{Z}_2$}
As aforementioned, we begin by discussing the group properties of $G$ and $H$, including their linear and projective representations. 

\subsubsection{Subgroup Properties of $\mathds{Z}_4\times \mathds{Z}_2$}
\label{sec:ExampleGroupProperties}
In this section, we discuss how the subgroup, $\tilde{H}$ fits into the structure of $G$. In doing so, we will introduce a detailed notation that will allow us to fully generalize the following discussion in Appendix~\ref{sec:AbelianGroupProperties}.
Recall, we chose to consider the subgroup
\begin{align}
    \tilde{H} = 2 \mathds{Z}_2 \times \mathds{Z}_2 =  \{ (0,0), (0,1), (2,0), (2,1)  \} \le G.
    \label{eq:ExampleSubgroup}
\end{align}
It will be useful to also define
\begin{align}
    \label{eq:ExampleSubSet}
    H = \mathds{Z}_2 \times \mathds{Z}_2 =  \{ (0,0), (0,1), (1,0), (1,1)  \} \subseteq G.
\end{align}
Note, $H$ is only a subset of $G$. However, defining the group operation $(i,j)\oplus_H (k,l) = (i\oplus_2 j, k\oplus_2 l)$, for all $(i,j),(k,l)\in H$, we have $(\tilde{H},\oplus_G) \cong (H, \oplus_H)$.

As discussed in Appendix~\ref{app:noninjectivesym}, our choice of subgroup, $\tilde{H}$, implicitly defines a permutation action of $U_g$. In particular, the permutation action is constant on the cosets of $\tilde{H}$. In the case of abelian groups, all subgroups are normal~\footnote{A subgroup is normal if it is invariant under conjugation by all members of the group. It is a special property of abelian groups that all subgroups are normal.}. Consequently, the set of cosets can be promoted to the quotient group, and the permutation action of $U_g$ becomes a representation of this quotient group. Explicitly, we write the quotient group as 
\begin{equation}
    G/\tilde{H}=\{(0,0) \tilde{H}, (1,0) \tilde{H} \}
\end{equation}
which is isomorphic to 
\begin{align}
    K = \mathds{Z}_2 \times \mathds{1} = \{(0,0),(1,0)\} \subseteq G
\end{align}
equipped with $(i,0)\oplus_K (k,0) = (i\oplus_2 k, 0)$, for all $(i,0),(k,0)\in K$.

It will be useful to note that there is a natural bijection between $G$ and $H\times K$, $g\mapsto (h(g),k(g))$, corresponding to element-wise Euclidean division. Namely, let us define
\begin{align}
    h(g) &= \left( \left\lfloor \frac{i}{2} \right\rfloor , j \right) \label{eq:ExampleBijection1}\\
    k(g) &= ( \text{Mod}(i,2),0)  \label{eq:ExampleBijection2}
\end{align}
for all $g=(i,j)\in G$. 

The reverse map of this bijection is given by\footnote{Note that the image of $h$ can be viewed as elements of either the group $H$ or the group $G$ (likewise for $K$). Hence this expression is well-defined.}
\begin{align}
    g=(i,j) &= |K|\ h\left( g \right) \oplus_G k\left( g \right) \\
    &= (2,1) \left( \left\lfloor \frac{i}{2} \right\rfloor,\ j \right) \oplus_G  ( \text{Mod}(i,2),\ 0)\\
    &= \left(2 \left\lfloor \frac{i}{2} \right\rfloor \oplus_4 \text{Mod}(i,2),\ j\right)   \label{eq:ExampleEuclidianDivision}
\end{align}
where we have defined $ |K| = (2,1)$. Thus, we see that the bijection corresponds to element wise Euclidean division with respect to $|K|$, where $h(g)$ is the tuple of quotients and $k(g)$ the tuple of remainders. As a result of this structure, $h(g)$ and $k(g)$ obey a number of identities [see Eqs.~\eqref{eq:AbelianGroupIdenitity1}~to~\eqref{eq:AbelianGroupIdenitity6} in the next appendix].

Finally, the structure of finite abelian groups means that the quotient group is in fact isomorphic to another subgroup of $G$. Namely, $(K,\oplus_K)$ is isomorphic to $(\tilde{K}, \oplus_G)$ with 
\begin{equation}
    \tilde{K}=|H| K = \{(0,0),\ (2,0)\} = 2\mathds{Z}_2 \times 1 \le G,
\end{equation}
where $ |H| = (2,2)$. As with $\tilde{H}$, one can then take the quotient of $G$ by $\tilde{K}$
\begin{equation}
    G/\tilde{K} = \{(0,0) \tilde{K},\ (0,1) \tilde{K},\ (1,0) \tilde{K},\ (1,1) \tilde{K} \}.
\end{equation}
This yields a second natural bijection between $G$ and $(H,K)$
\begin{align}
    \hat{h}\left( (i,j) \right) &= \left( \text{Mod}(i,2), j  \right) \label{eq:ExampleBijection3}\\
    \hat{k}\left( (i,j) \right) &= \left( \left\lfloor \frac{i}{2}\right\rfloor , 0 \right).\label{eq:ExampleBijection4}
\end{align}
Note, under this bijection, elements of $H$ now correspond to the remainders and elements of $K$ correspond to the quotients. 

\subsubsection{Linear Representation Properties of $\mathds{Z}_4\times \mathds{Z}_2$}
\label{sec:ExampleAbelianGroupLinearRep}
Next, we discuss the linear representations of $G$ and $H$. As $G$ is abelian, the irreps of $G$ are 1D and can be labeled by elements of $G$. Specifically, we may write 
\begin{align}
    \chi^{q\in G}(g) = \chi^{q=(c,d)\in G}_{g=(a,b)\in G} 
    &=  ((i)^c)^a \times ((-1)^d)^b,
    \label{eq:ExampleCharacterofG}
\end{align}
where $q\in G$ labels the irrep and $\chi^q_g$ is therefore a phase. These can be understood as powers of roots of unity, in which $q\in G$ specifies the root of unity, and $g$ specifies to which power this root should be raised.

One can also consider the irreps of the subgroup $\tilde{H}\cong H$
\begin{align}
    \tilde{\chi}^{p=(c,d)\in H}_{h=(a,b)\in H} &= ((-1)^c)^a \times ((-1)^d)^b
\end{align}
Clearly, the irreps of $H$ are contained in the irreps of $G$. 

Explicitly, we may write
\begin{align}
    \tilde{\chi}^{(c,d)\in H}_{(a,b)\in H} &= \chi^{(2c, d)\in G}_{(a,b)\in G}.  \label{eq:Exampleirrepsubgroup}
\end{align}
But we can also reference irreps of $G$ using elements of $H$, e.g.,
\begin{equation}
    \chi^{(1,0)\in H\subseteq G}_{g=(a,b)} = i^a. \label{eq:ExampleCharacterUsingSubset}
\end{equation}
Here we are considering $(1,0)\in H$ as an element of $G$. We will find this notation useful later when defining $U_g$ and $P_{r,q}$ (see also Appendix~\ref{sec:AbelianGroupLinearRep}).

\subsubsection{Projective Representations of $\mathds{Z}_2\times \mathds{Z}_2$}
\label{sec:ExampleProjectiveRepresentations}

We now consider projective representations of the subgroup $\tilde{H}\cong H=\mathds{Z}_2\times \mathds{Z}_2$. As mentioned, it is well known that $\mathds{Z}_2\times \mathds{Z}_2$ has a second cohomology group isomorphic to $\mathds{Z}_2$. Moreover, finite abelian groups have the property that, up to projective equivalence, there is only one irreducible projective representation per cohomology class. The trivial projective representation, $\mu=0$, can be chosen as $\omega_{g}^{0}=1$ for all $g\in H$. For the nontrivial class, $\mu=1$, it is well-known that the irreducible projective representation can be chosen to be the Pauli representation:
\begin{align}
    \omega_{(0,0)\in H}^{1}&= \id, & \omega_{(1,0)\in H}^{1}&= \sigma_z, \nonumber\\
    \omega_{(0,1)\in H}^{1}&= \sigma_x, & \omega_{(1,1)\in H}^{1}&=\sigma_y. \label{eq:pauliprepresentations}
\end{align}
In the following, we consider the nontrivial cohomology class, $\mu=1$, and therefore drop the $\mu$ from now on. It is easily verified by considering commutation relations that
\begin{equation}
     \omega_g \omega_h = \tilde{\chi}^{\phi(h)}_g \omega_h \omega_g,\ \forall g,h\in H,
     \label{eq:ExamplePauliQuasiCom}
\end{equation}
where $\tilde{\chi}^{\phi(h)}$ is an irrep of $H$ labeled by $\phi: H \rightarrow H$, which is given by 
\begin{equation}
    \phi ((i,j))  = (j,i). \label{eq:ExampleHomomorphism}
\end{equation}
Later, we will use the fact that $\phi: H \rightarrow H$  is in fact a homomorphism. One might also note that $\{\omega_{g}\}$ forms an orthonormal basis for $2\times2$ matrices. This also is not a coincidence as we will see later in Appendix~\ref{sec:AppendixProofProjectiveRepresentations}.

\subsection{Example of Symmetry and Measurements}
Let us finish this appendix by explicitly providing the symmetry in Eq.~\eqref{eq:Ug} and the measurements in Eq.~\eqref{eq:Projectors} using the notation introduced above.

To begin, the symmetry we chose in Eq.~\eqref{eq:Ug} is given by
    \begin{align}
    U_{g=(i,j)} &=\left(\bigoplus_{\alpha\in K}  (\omega_{h(g,\alpha)}^\mu)^*\otimes (\omega_{h(g,\alpha)}^\mu)\right) \left( X_{|K|}^{k(g)} \otimes \id^{\otimes 2}\right) \nonumber \\
    & = \left(\bigoplus_{a \in \{0,1\}} \left(\sigma_x^{j} \sigma_z^{\lfloor \frac{i\oplus_4 a}{2} \rfloor} \right)^{\otimes 2}  \right) \left(\sigma_x^{\text{Mod}(i,2)} \otimes \id^{\otimes 2}\right)
\end{align} 
Note, as $|K|=2$, the cyclic permutation is simply a Pauli-x. Additionally, note that this is not the same as the original symmetry\footnote{Strictly speaking, using the notation of Eq.~\eqref{eq:UgBlockStructure}, the above equation is for $W_g$}. Nonetheless, it is easy to verify that this symmetry appears as a sub-representation of $U_g^{\otimes4}$. Thus, following the discussion in Appendix \ref{sec:BlockingAndLocalSymmetry}, we consider all actions after blocking four sites together.

Next, we have the measurements. As explained, we project onto $(\id\otimes V_{r,q})\ket{\tilde{\Phi}^+}$, where now
\begin{align}
    \ket*{\widetilde{\Phi}^+}&\propto \sum_{i\in\{0,1\}} \sum_{j,k\in\{0,1\}} \ket{i,j,k,i,k,j}
\end{align}
and
\begin{equation}
    V_{r,q} = U_r \tilde{V}_q\equiv U_r\  [\tilde{Z}^q \otimes  \id \otimes (\omega_{h(q)}^\mu)],
\end{equation}
where
\begin{align}
 \tilde{Z}^{q=(c,d)} &= Z_{|K|}^{k(q)} \left( \bigoplus_{\alpha\in K} \chi^{\phi_\mu (h(q)) }_\alpha   \right) \label{eq:ExampleAbstract}\\
            &= \sigma_z^{\text{Mod}(c,2)} 
                \begin{pmatrix}
                    1 & 0 \\
                    0 & (i)^d   
                \end{pmatrix} \label{eq:ExampleZTilde}
\end{align}
In Eq.~\eqref{eq:ExampleAbstract}, we decompose the phase matrix into a generalized Pauli-z (which, as $|K|=2$, is in fact just Pauli-z) and a phase matrix. This second phase matrix ensures that $\tilde{V}_q$ quasi-commutes with $U_g$ when $h(g)\ne e$. Note in Eq.~\eqref{eq:ExampleAbstract}, we are referencing an irrep of $G$ using an element of $H\subseteq G$, as discussed in Appendix \ref{sec:ExampleAbelianGroupLinearRep}.

Given this definition of $U_g$ and $V_{r,q}$, it is easy to verify that $V_{r,q}$ quasi-commutes with $U_g$ and forms an orthonormal unitary matrix basis. Moreover, it is easy to verify that Eq.~\eqref{eq:slidethrough} and Eq.~\eqref{eq:ZeroVectorOutcome} hold. Consequently, the protocol as described in the main text is symmetry preserving and the phase $(2\mathbb{Z}_2\times \mathbb{Z}_2, 1)$ is reachable from the trivial phase. In the next appendix, we generalize these constructions for all phases of all finite abelian groups and show the same properties hold.

\section{Proof of Result~\ref{res:AbelianPhaseDiagram}}\label{app:abelian}
In this section, we prove Lemmata~\ref{lem:bellmeasurements}, and~\ref{lem:ZeroVectorOutcome}, and consequently Result~\ref{res:AbelianPhaseDiagram} in the main text. As some of the notation can be dense, we remind the reader that one can always consult the detailed example given in Appendix~\ref{app:Example}.

\subsection{Properties of Finite Abelian Groups}
As before, we begin by recapping some general properties of finite abelian groups and their subgroups, including linear and projective representations. 

\subsubsection{Group Properties}
\label{sec:AbelianGroupProperties}
To begin, by the Fundamental Theorem of Finite Abelian Groups, any finite abelian group is isomorphic to a product of cyclic groups of prime power order, i.e., 
\begin{equation}\label{eq:GroupStructure}
    G \cong \times_{m\in[M]} \mathds{Z}_{p_m^{r_m}} \equiv \times_{m\in[M]} G_m
\end{equation}
with $p_m,\ r_m,\ M \in \mathds{N}$ and $p_m$ are prime. We will sometimes write the corresponding group elements of $G$ as $g=(g_m)_{m\in[M]}$. We use $\oplus_G$ to denote element-wise modulo addition, i.e, $g_1\oplus_G g_2 = ((g_1)_m\oplus_{|G_m|}(g_2)_m)_{m\in[M]}$ [see Eq.~\eqref{eq:ExampleGroupAddition}].

For abelian groups, any subgroup, $\tilde{H}\le G$, is also isomorphic to a product of cyclic groups of the-same-prime power order, i.e., $\tilde{H} \cong \times_{m\in[M]} \mathds{Z}_{p_m^{\tilde{r}_m}}\equiv H$ with $\tilde{r}_m\le r_m$. Moreover, every subgroup is normal. Therefore, given a subgroup $\tilde{H}\le G$, the quotient group $G/\tilde{H}$ is well-defined. For $\tilde{H}\le G$ as above, we define
\begin{equation}
    G/\tilde{H} \cong \times_{m\in[M]} \mathds{Z}_{p_m^{r_m-\tilde{r}_m}}\equiv K \equiv \times_{m\in[M]} K_m.
    \label{eq:QuotientIsomorphism}
\end{equation} 

Note, in the above notation, we use $\tilde{H}$ to denote the subgroup of $G$ (equipped, therefore, with the group operation of $G$), i.e., $\tilde{H}\le G$. On the other hand, $H$ is in fact a sub\textit{set} of $G$, i.e., $H\subseteq G$ [see Eq.~\eqref{eq:ExampleSubgroup} and Eq.~\eqref{eq:ExampleSubSet}]. We can promote $H$ to a group by equipping it with its own group multiplication, $\oplus_H$, in which case we have $(\tilde{H},\oplus_G)\cong (H,\oplus_H)$. Similarly, $K$ is also a subset of $G$, i.e., $K\subseteq G$, but we can equip it with $\oplus_K$ to promote it to a group, yielding $(G/H, \oplus_G)\cong (K,\oplus_K)$.  We make this distinction between the subset $H$ and the subgroup $\tilde{H}$ as we will need the subset $H$ for the constructions $U_g$ and $P_{r,q}$. For a concrete example of this convention, see Appendix~\ref{sec:ExampleGroupProperties}.

Given the quotient group, there is then a natural bijection, $g\in G \mapsto (h(g),k(g))\in (H,K)$, corresponding to element-wise Euclidean division [see Eqs.~\eqref{eq:ExampleBijection1},~\eqref{eq:ExampleBijection2},~and~\eqref{eq:ExampleEuclidianDivision}]. Generally, only $k(g)$ is a homomorphism. Indeed, we have the following identities
\begin{align}
    k(c)&=c \label{eq:AbelianGroupIdenitity1}, \\
    k(g_1 \oplus_G g_2) &= k(g_1)\oplus_K k(g_2) \label{eq:AbelianGroupIdenitity2}, \\
    h(g_1 \oplus_G g_2) &= h(g_1)\oplus_H h(k(g_1) \oplus_G g_2) \label{eq:AbelianGroupIdenitity3},\\
    g&= |K| h(g) \oplus_G k(g) \label{eq:AbelianGroupIdenitity4},
\end{align}
for all $g_1,g_2\in G$ and $c\in K\subseteq G$, and where $\oplus_G, \oplus_H$, and  $\oplus_K$ refer to modulo-addition in G, H and K respectively. Moreover, $|K|$ should be understood as the vector $|K|=(|K_m|)_{m\in[M]}$, and multiplication is element-wise (again, see Appendix~\ref{sec:ExampleGroupProperties} for a concrete example).

Finally, we can consider the elements in $H\subseteq G$ as elements in $G$ via two further identities. For this, we first recognize that there is another subgroup of $G$ isomorphic to $K$, $\tilde{K}\equiv |H|K \le G$\footnote{Note, this is another consequence of the Fundamental Theorem of Abelian Groups; that every quotient group of an abelian group is isomorphic to some subgroup of $G$. This is a special property that is not true for general groups.}. As with $\tilde{H}$, one may also quotient G by $\tilde{K}$. This yields a second, different natural bijection $g\in G\mapsto (\hat{h}(g),\hat{k}(g))\in (H,K)$ [see Eq.~\eqref{eq:ExampleBijection3} and~\eqref{eq:ExampleBijection4}]. Equipped with this bijection, we can then write
\begin{align}
    h_1 \oplus_G h_2 &= (h_1 \oplus_H h_2) \oplus_G |H|\ \hat{k}(h_1 \oplus_G h_2) ,\label{eq:AbelianGroupIdenitity5}\\
    h_1 \ominus_G h_2 &= (h_1 \ominus_H h_2) \oplus_G |H|\ \hat{k}(h_1 \ominus_G h_2),
    \label{eq:AbelianGroupIdenitity6}
\end{align}
for all $h_1,h_2\in H\subseteq G$.

\subsubsection{Linear Representations}
\label{sec:AbelianGroupLinearRep}
All irreps of an abelian group, $\chi^{q}:G\rightarrow U(1)$, are one dimensional. As the set of irreps of a finite abelian group, $G^*$, is isomorphic to the group itself, we can choose to the label irreps of $G$ by the elements of $G$ itself. Explicitly, for the group given in Eq.~\eqref{eq:GroupStructure}, the irrep labeled by $q\in G$ is given by
\begin{align}
    \chi^q(g)=\chi^q_g &= \prod_{m\in [M]} \exp(2\pi i\ \frac{g_m q_m}{|G_m|})\\
    &= \chi^g_q
\end{align}
[see also Eq.~\eqref{eq:ExampleCharacterofG}].

In the following, we will also consider irreps of subgroups. In particular, let $\tilde{\chi}$ correspond to irreps of $(H ,\oplus_H)$. Then we have
\begin{equation}
    \tilde{\chi}^{p\in H}_{h\in H} = \chi^{|K| p\in G}_{h\in H \subseteq G} = \chi^{p\in G}_{|K|h \in G}.
    \label{eq:IrrepsOfSubgroups}
\end{equation}
That is, we can relate irreps of $H$ to irreps of $G$ [see Eq.~\eqref{eq:Exampleirrepsubgroup}].
We introduce this notation as in order to define our physical symmetry later on, we will need to reference irreps of $G$ using elements of $H\subseteq G$ [see Eq.~\eqref{eq:ExampleCharacterUsingSubset}]. For a concrete example of this notation, see Appendix~\ref{sec:ExampleAbelianGroupLinearRep}.

\subsubsection{Projective Representations}
\label{sec:AppendixProofProjectiveRepresentations}
Here we introduce several properties of projective representations for finite abelian groups that will be useful in the following derivation (see Appendix~\ref{app:noninjectivesym} for general properties of projective representations). As in our subsequent arguments we only consider projective representations of the subgroup, we will consider here projective representations of $H$. 

So, to begin, let  $\omega_{g\in H}^\mu$ be the $\mu$-irrep\footnote{For definition of $\mu$-irrep, see Appendix~\ref{app:AppendixAInjectiveMPS}, and note, for abelian groups, projective irreps are unique up to projective equivalence.} obeying the relation
\begin{equation}
    \omega_g^\mu \omega_h^\mu = \gamma^\mu(g,h)\ \omega_{gh}^\mu,
\end{equation}
for all $g,h\in H$, where $\gamma^\mu: H \times H \rightarrow U(1)$ is the corresponding cocycle. One can show that, for any $g\in H$, $\tilde{\gamma}^{\mu,g}: H\mapsto U(1)$, given by
\begin{equation}
    \tilde{\gamma}^{\mu,g}(x)= \frac{\gamma^{\mu}(x,g)}{\gamma^{\mu}(g,x)},
\end{equation}
is in-fact a homomorphism between $H$ and the irreps of $H$ \cite{Kleppner1965, ElseEtAl2012_MBQC1}. That is, there is some homomorphism, $\phi_\mu: H \rightarrow H$ such that 
\begin{equation}
    \tilde{\gamma}^{\mu,g}(x)= \chi^{\phi_\mu(g)}_x
\end{equation}
[see Eq.~\eqref{eq:ExampleHomomorphism}]. Moreover, we have
\begin{equation}
    \omega_g^{\mu} \omega_h^{\mu} = \chi^{\phi_\mu(h)}_g \omega_h^{\mu} \omega_g^{\mu},\ \forall g,h\in H.
    \label{eq:ProjectiveQuasiComRel}
\end{equation}

We may also define the \emph{projective center} of $H$ with respect to $\mu$,
\begin{equation}
    Z^{\mu}(H)= \{g\in H : \phi_\mu(g)=e \}\le H.
\end{equation}
That is, $Z^{\mu}(H)$ is the set of $g\in H$ for which $\tilde{\gamma}^{\mu,g}$ is the trivial representation. When $Z^{\mu}(H)=\{e\}$, $\mu$ is referred to as a maximally non-commuting (MNC) phase. It is easily verified that the Pauli representation from the example in Appendix~\ref{sec:ExampleProjectiveRepresentations} corresponds to one of these MNC phases. Indeed, these phases have found considerable applications in measurement-based quantum computing~\cite{ElseEtAl2012_MBQC1, StephenEtAl2017_MBQC2, Stephen2019_MBQC3}.

Given $Z^{\mu}(H)$ the dimension of the $\mu$-irrep is then given by~\cite{karpilovsky1994_ProjTraceCondition}\footnote{See the Pauli representation in Appendix~\ref{sec:ExampleProjectiveRepresentations}. $|H|=|\mathds{Z}_2\times \mathds{Z}_2|=4$ and the projective center is trivial. Therefore, $D_\mu = \sqrt{4/1}=2$.}
\begin{equation}
    D_\mu =\sqrt{\frac{|H|}{|Z^{\mu}(H)|}}.
    \label{eq:ProjRepDim}
\end{equation}
Finally, for any $\mu$-irrep, the following holds \cite{karpilovsky1994_ProjTraceCondition}
\begin{equation}
        \mathrm{tr}[\omega_h^{\mu}] = \begin{cases}
            D^\mu, & \text{if }h\in Z^{\mu}(H)\\
            0, &  \text{if }h\not\in Z^{\mu}(H).
        \end{cases}
        \label{eq:ProjRepsInComCenterAreOrthogonal}
\end{equation}
This yields the following corollary.
\begin{cor}\label{corr:basis}
    Consider an $\mu$-irrep $\omega^{\mu}_{g\in H}\in U(D_{\mu})$. Let $Q \subseteq H$ be a set of representatives for each of the cosets $\frac{H}{Z^{\mu}(H)}$. Then $\{\omega^{\mu}_g\}_{g\in Q}$ is an orthonormal basis for $D^{\mu} \times D^{\mu} $ complex matrices.
\end{cor}

\begin{proof}
    Firstly, note that $|Q|=|\frac{H}{Z^{\mu}(H)}|=(D^\mu)^2$. Secondly, 
    \begin{align}
        \mathrm{tr} \left[ (\omega_{x}^{\mu})^\dagger \omega_{y}^{\mu} \right]&\propto \tr \left[ \omega_{x^{-1}y}^\mu \right]\\
        &=  \begin{cases}
            D^\mu, & \text{if }x^{-1}y \in Z^{\mu}(H)\\
            0, &  \text{if }x^{-1}y \not\in Z^{\mu}(H)
        \end{cases}\\
        &=  \begin{cases}
            D^\mu, & \text{if }y \in x Z^{\mu}(H)\\
            0, &  \text{if } y \not\in x Z^{\mu}(H)
        \end{cases}\\
         &=  \begin{cases}
            D^\mu, & \text{if } x=y \\
            0, &  \text{if } x\ne y
        \end{cases}
    \end{align}
    for all $x,y\in Q$, where the last line follows as there is a unique element in Q per coset. Thus $\{\omega^{\mu}_g\}_{g\in Q}$ is an ONB.
\end{proof}
Note that if $\mu$ corresponds to a MNC class, then $\{\omega_g^\mu\}_{g\in H}$ is an ONB for $\mathds{C}_{D^\mu, D^\mu}$ (see also~\cite{Klappenecker2002}). That is the full set, $\{\omega_g^\mu\}_{g\in H}$, is an ONB rather than a strict subset. Once again, for a concrete example of any of the introduced notation, see Appendix~\ref{sec:ExampleProjectiveRepresentations}.

\subsection{Symmetries and Measurements}
Given the above notation, we now precisely define, in full generality, the physical symmetry $U_g$, referred to in Eq.~\eqref{eq:Ug}, and the symmetric, projective measurement, $\{P_{r,q}\}_{r,q\in S}$ referred to in Eq.~\eqref{eq:projMeasurement}.

\subsubsection{Symmetry}
We begin with the physical symmetry. Recall, given the phase $(H,\mu)$, in Eq.~\eqref{eq:Ug}, we had
\begin{equation}
\begin{aligned}\label{fig::AppendixPhaseRepresentativeSymmetry}
    \includegraphics[width=0.45\linewidth]{PhaseRepresentativeSymmetry.pdf}
\end{aligned}
\end{equation}
Equivalently, we can write this as
\begin{align}
    U_g &=\left(\bigoplus_{\alpha\in K}  (\omega_{h(g,\alpha)}^\mu)^*\otimes (\omega_{h(g,\alpha)}^\mu)\right) \left(P_{k(g)} \otimes \id \otimes \id\right)\\
    &=\left[\bigoplus_{\alpha\in K} W_{h\left(g\oplus_G\alpha\right)}^\mu \right] \left(P_{k(g)} \otimes \id \otimes \id\right)
    \label{eq:AppendixPhysicalSymmetry}
\end{align}
where we have defined $W_{h(g, \alpha)}^\mu = (\omega_{h(g,\alpha)}^\mu)^*\otimes (\omega_{h(g,\alpha)}^\mu)$ to simplify notation\footnote{Note, that $W_{h}^\mu$ is a linear representation of $H$.}. Explicitly, we choose the permutation representation to be given by
\begin{align}
    P_{k(g)}=\otimes_{m\in[M]} X_{|K_m|}^{k_m(g)}\equiv X^{k(g)}_{|K|}
\end{align}
where $X_{|K_i|}$ is $|K_i|$-dim shift operator, i.e., $X_{|K_i|}=\sum_{i\in[|K_i|]} \ket{i}\bra{i\oplus 1}$\footnote{The shift operator here generalizes the Pauli X operator from the example in Appendix~\ref{app:Example} to higher dimensions. It is easily verified that this satisfies the properties of $P_g$ as discussed in Appendix~\ref{app:noninjectivesym} [see Eq.~\eqref{eq::PermutationAction}].}. Additionally, $(\omega^\mu)_h$ is the $\mu$-irrep of $H$ and $h(g,\alpha)=h(g\oplus_G \alpha)$, where $\alpha\in K\subseteq G$ is considered as an element in G and $\oplus_G$ denotes group addition in $G$. Note, that
\begin{equation}
d\equiv\text{dim}(U_g)=|K| D_\mu^2= \frac{|G|}{|H|} \frac{|H|}{|Z^\mu(H)}= \frac{|G|}{|Z^\mu(H)|}.   
\end{equation}
In particular, if $\mu$ corresponds to a MNC phase (of $H$), then $d=|G|$. Finally, one can verify that $U_g$ is indeed a linear representation.

\subsubsection{Measurements}
Recall that we defined the following measurement operators in Eq.~\eqref{eq:projMeasurement} in the main text
\begin{align}\label{eq:AppendixprojMeasurement}
     P_{r,q} &= (\id \otimes V_{r,q}) \ketbra*{\widetilde{\Phi}^+}(\id \otimes V_{r,q}^\dagger),
\end{align}
where
\begin{equation}
    \ket*{\widetilde{\Phi}^+}\propto\sum_{i\in [|K|]} \sum_{j,k\in [D_\mu]} \ket{i,j,k,i,k,j}.
\end{equation}
We emphasize again that the state $\ket*{\widetilde{\Phi}^+}$ is different from the typical $\ket{\Phi^+}$ state as we must ensure the legs connect appropriately. This can be seen for instance from the following graphical representation of the measurement operators.
\begin{equation*}
\begin{aligned}
    \includegraphics[width=0.35\linewidth]{MeasureOp.pdf}
\end{aligned}
\end{equation*}
Let us now describe in detail the structure of the unitary $V_{r,q}$. It is given by $V_{r,q} = U_r \tilde{V}_q$, where $r,q\in G$, and thus, is a product of the physical symmetry $U_r$, and a unitary $\tilde{V}_q$. The latter is given by
\begin{align}\label{eq:appendixV}
 \tilde{V}_{q\in G} \equiv \tilde{Z}^q \otimes  \id \otimes (\omega_{h(q)}^{\mu}).
\end{align}
The last tensor factor, $\omega_{h(q)}^{\mu}$, is the same projective representation as in $U_g$, and the first tensor factor, $\tilde{Z}^q$, is a diagonal phase matrix. This phase matrix is itself a product of two phase matrices
\begin{equation}
    \tilde{Z}^q = \underbrace{\vphantom{\left( \bigoplus_{\alpha\in K} \chi^{\phi_\mu (h(q)) }_\alpha   \right)}Z_{|K|}^{k(q)}}_{(1)} \underbrace{\left( \bigoplus_{\alpha\in K} \chi^{\phi_\mu (h(q)) }_\alpha   \right)}_{(2)}.
\end{equation}
The first phase matrix is defined by
\begin{equation}
    Z^{k(q)}_{|K|}\equiv \bigotimes_{m\in [M]}Z^{k_m(q)}_{|K_m|},
\end{equation}
where $Z_{|K_m|}$ is the $|K_m|$-dim clock operator, i.e.,
\begin{equation}
Z_{|K_m|}=\sum_{j\in[|K_m|]} \exp(\frac{2 \pi i}{|K_m|}j)\ketbra{j}{j}.    
\end{equation}
In the second phase matrix $\chi_\alpha ^{\phi_\mu(h(q))}$ is a linear irrep of G, with $\alpha \in K \subseteq G$, and $\phi_\mu(h(q)) \in H \subseteq G$ coming from the quasi-commutation relation in Eq.~\eqref{eq:ProjectiveQuasiComRel}.

Let us now elaborate more on the multiplication properties of the matrices $\{\tilde{V}_{q}\}_{g}$, which will become useful later. It is important to note that $\{\tilde{V}_{q}\}_{g\in G}$ is closed under multiplication, up to phases, as the following lemma shows.
\begin{lem}\label{lem:Vclosed1}
Let $\tilde{V}_{p,q}$ be defined as above. Then,
\begin{align}
    \tilde{V}_p \tilde{V}_q &\propto \tilde{V}_{f(p,q)}, \\
    \tilde{V}_p^\dagger \tilde{V}_q &\propto \tilde{V}_{\tilde{f}(p,q)},
\end{align}
where $f: G\times G \rightarrow G$ is defined by
\begin{multline}
    f(p,q)=|K|\left(h(p)\oplus_h h(q)\right) \oplus_G  \\
    \quad\left\{ k(p)\oplus_K k(q) \oplus_K \hat{k}\left[\phi_\mu \left(h(p)\right) \oplus_G \phi_\mu \left(h(q)\right)\right] \right\},
\end{multline}
and $\tilde{f}: G\times G \rightarrow G$ is defined by
\begin{multline}
    \tilde{f}(p,q)=|K|\left(h(q)\ominus_h h(p)\right) \oplus_G  \\
    \quad \left\{ k(q)\ominus_K k(p) \oplus_K \hat{k}\left[\phi_\mu \left(h(q)\right) \ominus_G \phi_\mu \left(h(p)\right)\right] \right\}.
\end{multline}
\end{lem}
\begin{proof}
To see this, first note that
\begin{align}
    \tilde{Z}^p \tilde{Z}^q&= \bigoplus_{\alpha\in K\subseteq G} \hat{\chi}^{k(p)}_\alpha  \chi^{\phi_\mu (h(q)) }_\alpha  \hat{\chi}^{k(q)}_\alpha  \chi^{\phi_\mu (h(p)) }_\alpha  \\
    &= \bigoplus_{\alpha\in K\subseteq G}  \hat{\chi}^{k(p)\oplus_K k(q)}_\alpha \chi^{\phi_\mu (h(p)) \oplus_G \phi_\mu (h(q))}_\alpha\\
    &= \bigoplus_{\alpha\in K\subseteq G}  \hat{\chi}^{k(p)\oplus_K k(q)}_\alpha \chi^{\phi_\mu (h(p)) \oplus_H \phi_\mu (h(q))}_\alpha \nonumber\\
    &\qquad \qquad \qquad \times \chi^{|H|\hat{k}\left[\phi_\mu (h(p)) \oplus_G \phi_\mu (h(q))\right]}_\alpha\\
    &= \bigoplus_{\alpha\in K\subseteq G}  \hat{\chi}^{k(p)\oplus_K k(q) \oplus_K \hat{k}\left[\phi_\mu (h(p)) \oplus_G \phi_\mu (h(q))\right]}_\alpha \nonumber\\
    &\qquad \qquad \qquad    \times  \chi^{\phi_\mu (h(p)) \oplus_H \phi_\mu (h(q))}_\alpha\\
    &= \tilde{Z}^{f(p,q)},
\end{align}
where we have used the group identity in Eq.~\eqref{eq:AbelianGroupIdenitity5} in the 3rd line and the irrep identity in Eq.~\eqref{eq:IrrepsOfSubgroups} in the fourth line.
Returning to $\tilde{V}$, we have
\begin{align}
    \tilde{V}_p \tilde{V}_q &= \tilde{Z}^p\tilde{Z}^q \otimes  \id \otimes \omega_{h(p)}\omega_{h(q)}\\
    & \propto \tilde{Z}^{f(p,q)} \otimes \id\otimes \omega_{h(p)\oplus_H h(q)}  \\
    &= \tilde{Z}^{f(p,q)} \otimes \id \otimes \omega_{h(f(p,q))} \\
    &= V_{f(p,q)}
\end{align}
where we have used $h(f(p,q))=h(p)\oplus_H h(q)$. Thus, the set $\{\tilde{V}_{q}\}_{g\in G}$ is closed up to phases. By an identical argument, we also have
\begin{equation}
    \tilde{V}_p^\dagger \tilde{V}_q \propto V_{\tilde{f}(p,q)}.
\end{equation}
\end{proof}
Moreover, the following lemma holds.
\begin{lem}\label{lem:Vclosed2}
Let $\tilde{V}_{p,q}$ be defined as above. Then,
    \begin{align}
    \tilde{V}_p^\dagger \tilde{V}_q = V_0 = \id \Leftrightarrow p=q.
\end{align}
\end{lem}
\begin{proof}
To begin, $\tilde{f}(p,q)=0$ only if $h(q)=h(p)$. In this case, $\phi_\mu(h(q))=\phi_\mu(h(p))$ and thus $\hat{k}[\phi_\mu(h(q))\ominus_G\phi_\mu(h(p))]=\hat{k}(0) =0$. Thus $\tilde{f}(p,q)=0$ only if $k(q)=k(p)$. Therefore, $\tilde{f}(p,q)=0$ if and only if $h(q)=h(p)$ and $k(q)=k(p)$.
\end{proof}

\subsection{Proof of Lemma \ref{lem:bellmeasurements}}\label{app:LemmaBell}
Having explicitly defined $U_g$, and $P_{r,q}$, we now prove Lemma \ref{lem:bellmeasurements}. For clarity, let us begin by restating Lemma~\ref{lem:bellmeasurements} more precisely.
\begin{manuallemma}{3}
  Let $G$ be a finite abelian group and $(H,\mu)$ label a phase under $G$. Let $Q\in H$ be a set of representatives for the cosets $\frac{H}{Z^{\mu}(H)}$ and $S^\mu=\{g \in G : h(g)\in Q \}$. Correspondingly, let $U_g$ be given as in Eq.~\eqref{eq:AppendixPhysicalSymmetry} and let $\{P_{r,q}\}_{q,r\in S^\mu}$ be given as in Eq.~\eqref{eq:AppendixprojMeasurement}.
  Then  $\{P_{r,q}\}_{q,r\in S^\mu}$ is a complete, symmetric projective measurement. 
\end{manuallemma}
\begin{proof}
    We start by verifying that $P_{r,q}$ is symmetric for all $r,q\in G$. It is easily verified that \footnote{Note, that we get a $U^\dagger$ instead of a $U^T$ because of the structure of $\ket*{\widetilde{\Phi}^+}$.} 
    \begin{equation}
        (U_g\otimes U_g) (\id \otimes A) \ket*{\widetilde{\Phi}^+}= (\id \otimes U_g A U_g^\dagger)\ket*{\widetilde{\Phi}^+}
    \end{equation}
    for any operator A. Thus it is clear that $[U_g^{\otimes 2}, P_{r,q}]=0,\ \forall g\in G$ if and only if $U_g \tilde{V}_{q} \propto \tilde{V}_{q} U_g,\ \forall g\in G$. We can verify this as follows. Consider
    
      \begin{align}
        U_g \tilde{V}_q &= \left[\bigoplus_{\alpha\in K} W_{h\left(g\oplus_G\alpha\right)}^\mu \right] \left(X^{k(g)} \otimes \id \otimes \id\right) \nonumber \\
        & \qquad \left(   Z_{|K|}^{k(q)} \left[ \bigoplus_{\alpha\in K} \chi^{\phi_\mu (h(q)) }_\alpha    \right] \otimes \id \otimes  (\omega_{h(q)}^{\mu})\right).
      \end{align}
      First, we move the clock matrix through the shift matrix, yielding a global phase which is a irrep of $K$. The clock matrix then commutes with the projective representations, and therefore, we can move it all the way to the front. This results in
      
        \begin{multline}
        \hat{\chi}^{k(q)}_{k(g)} \left(Z^{k(q)}_{|K|} \otimes \id \otimes \id \right)  \left[\bigoplus_{\alpha\in K} W_{h\left(g\oplus_G\alpha\right)}^\mu \right] \\
        \left(X^{k(g)} \otimes \id \otimes \id\right)  \left(  \left[ \bigoplus_{\alpha\in K} \chi^{\phi_\mu (h(q)) }_\alpha    \right]\otimes \id \otimes  (\omega_{h(q)}^{\mu})  \right).
        \end{multline}
        Next, we move the rest of the diagonal phase matrix through the shift matrix. This will permute the diagonal elements. It can then be brought to the front, yielding
        \begin{multline}
        \hat{\chi}^{k(q)}_{k(g)} \left(Z^{k(q)}_{|K|} \left[\bigoplus_{\alpha\in K}  \chi^{\phi_\mu (h(q)) }_{\alpha\oplus_K k(g)}  \right] \otimes \id \otimes \id \right) \\
        \left[\bigoplus_{\alpha\in K} W_{h\left(g\oplus_G\alpha\right)}^\mu \right]  \left(X^{k(g)}\otimes \id  \otimes (\omega_{h(q)}^{\mu}) \right).
        \end{multline}
        Now we move the projective representation part of $\tilde{V}$ through to the left. This picks up a phase in each block according to Eq.~\eqref{eq:ProjectiveQuasiComRel} when passing through $\oplus_\alpha W_{h(g\oplus\alpha)}$, yielding
        \begin{multline}
        \hat{\chi}^{k(q)}_{k(g)} \left(Z^{k(q)}_{|K|} \left[\bigoplus_{\alpha\in K}  \chi^{\phi_\mu (h(q)) }_{\alpha\oplus_K k(g)} \tilde{\chi}^{\phi_\mu(h(q))}_{h(g\oplus_G \alpha)}  \right] \otimes \id \otimes (\omega_{h(q)}^{\mu})  \right) \\
        \left[\bigoplus_{\alpha\in K} W_{h\left(g\oplus_G\alpha\right)}^\mu \right]   \left(X^{k(g)} \otimes \id \otimes \id\right).
         \label{eq:AppendixLemma2Derivation1}
        \end{multline}
        The phase we picked up in each block, $\tilde{\chi}$, is expressed as irrep of $H$. We can convert it to an irrep of $G$, $\chi$, via the following equation
        \begin{equation}
             \tilde{\chi}^{\phi_\mu(h(q))}_{h(g\oplus_G \alpha)}= \chi^{\phi_\mu(h(q))}_{|K| h(g\oplus_G \alpha)}.
        \end{equation}
        Using this, and  the fact $\chi^{\phi_\mu(h(q))}$ is a homomorphism of $G$, we have
        \begin{align}
            \chi^{\phi_\mu (h(q)) }_{\alpha\oplus_K k(g)} \tilde{\chi}^{\phi_\mu(h(q))}_{h(g\oplus_G \alpha)}&=
            \chi^{\phi_\mu (h(q)) }_{\alpha\oplus_K k(g)} \chi^{\phi_\mu(h(q))}_{|K| h(g\oplus_G \alpha)}\\
            &=\chi^{\phi_\mu (h(q)) }_{(\alpha\oplus_K k(g)) \oplus_G |K| h(g\oplus_G \alpha)}.
        \end{align}
        Then, using the group identity Eq.~\eqref{eq:AbelianGroupIdenitity4}, one notices that $(\alpha\oplus_K k(g)) \oplus_G |K| h(g\oplus_G \alpha)=g\oplus_G \alpha$, and then
        \begin{equation}  
            \chi^{\phi_\mu (h(q)) }_{g\oplus_G \alpha}= \chi^{\phi_\mu (h(q)) }_{g} \chi^{\phi_\mu (h(q)) }_{\alpha}.
        \end{equation} 
        Inserting this into Eq.~\eqref{eq:AppendixLemma2Derivation1}, one finds that
        \begin{equation}
        U_g\tilde{V}_q= \hat{\chi}^{k(q)}_{k(g)} \chi^{\phi_\mu (h(q)) }_{g} \tilde{V}_q U_g.
        \end{equation}
        
    Thus we have verified that $U_g$ quasi-commutes with $V_q$ for all $g,q\in G$, and thus all projectors, $P_{q,r}$, are symmetric.
    
    We now verify the second part of the Lemma, namely that $\{P_{r,q}\}_{r,q\in S}$ is a complete measurement on $\mathds{C}_d^{\otimes 2}$. It is clear that this holds if and only if $\{V_{r,q}\}_{r,q\in S}$ is an orthonormal basis for $\mathds{C}_{d\times d}$. Recall, $Q\subseteq H$ is a set of representatives for the cosets $\frac{H}{Z^{\mu}(H)}$ and 
    \begin{equation}
        S=\{g \in G : h(g)\in Q \}.
    \end{equation}
    Note that $|S|= |K| |Q| = |K|(D^\mu)^2= d \equiv \text{dim}(U_g)$. Therefore, $|\{V_{r,q}\}_{r,q\in S}|=d^2$. Thus all that remains to be shown is that the $\{V_{r,q}\}_{r,q\in S}$ are orthonormal. As $U_r$ and $\tilde{V}_q$ quasi-commute, and are closed up to phases (see Lemmata~\ref{lem:Vclosed1}, and~\ref{lem:Vclosed2}), it is sufficient to show that $\tr[U_{r}^\dagger V_{q}]=\delta_{r,0}\delta_{q,0}\ d$.
    
Inserting Eqs.~\eqref{eq:AppendixPhysicalSymmetry}~and~\eqref{eq:appendixV} one obtains
\begin{multline}
    \tr[U_{r}^\dagger V_{q}] = \underbrace{\sum_{\alpha,\beta \in K} \tr[(X^{k(r)})^T\ketbra{\alpha}{\alpha} Z^{k(p)}  \chi^{\phi_\mu (h(q)) }_\beta  \ketbra{\beta}{\beta}]}_{(1)} \\
    \otimes \underbrace{\tr[\omega_{h(r\oplus \alpha)}^T]}_{(2)} \otimes  \underbrace{\tr[ \omega_{h(r \oplus \alpha)}^\dagger (\omega_{h(q)})]}_{(3)}. \label{eq:appOrthogEq}
\end{multline}
Evaluating the sum over $\beta$, the first term becomes
\begin{multline}\label{eq:random}
    \sum_{\alpha \in K}  \chi^{\phi_\mu (h(q)) }_\alpha  \bra{\alpha} Z^{k(q)} (X^{k(r)})^T \ket{\alpha}\\
    =\delta_{k(r),0}\sum_{\alpha \in K}  \chi^{\phi_\mu (h(q)) }_\alpha  \bra{\alpha}Z^{k(q)} \ket{\alpha}.
\end{multline}
In the second term, using Eq.~\eqref{eq:AbelianGroupIdenitity3}, one finds that $h(r\oplus\alpha)=h(r)\oplus_H h(k(r)\oplus_G \alpha)$. As the first term is proportional to $\delta_{k(r),0}$, and $h(\alpha)=0$ for $\alpha\in K$, we get $h(r\oplus\alpha)=h(r)$ and thus the second term in Eq~\eqref{eq:appOrthogEq} evaluates to $D_\mu\delta_{h(r),0}$. This implies that the third term in Eq~\eqref{eq:appOrthogEq} evaluates to $D_\mu\delta_{h(q),0}$. Due to $\delta_{h(q),0}$ one finds that $\chi^{\phi_\mu (h(q)) }_\alpha=1$ as $\phi_\mu$ is a homomorphism. Then, the sum in Eq.~\eqref{eq:random} simplifies to $\abs{K}\delta_{k(q),0}$. Finally, we arrive at
\begin{align}
    \tr[U_{r}^\dagger V_{q}] &=D_\mu^2 \abs{K}\delta_{k(r),0}\delta_{k(q),0}\delta_{h(q),0}\delta_{h(r),0}\notag \\
    &=D_\mu^2 \abs{K}\delta_{r,0}\delta_{q,0}.
\end{align}
Thus $\{V_{r,q}\}_{r,q\in S}$ does indeed form an orthonormal basis that quasi-commutes with $U_g$, and thus $\{V_{r,q}\}_{r,q\in S}$ forms a symmetric von Neumann measurement.
\end{proof}

\subsection{Proof of Eq.~\eqref{eq:slidethrough}}
\label{sec:AbelianSlideThrough}
In this section, we prove the rules in Eq.~\eqref{eq:slidethrough} in the main text, that describe how errors, induced by measurements, propagate in the target MPS. The first rule directly follows from the symmetry of the state with respect to $U_g^{\otimes 3}$. So let us prove the second rule. Let $V_q$ be defined as previously. Then the following identity holds.
\begin{equation}\label{eq:AppendixSlideThrough}
    \begin{aligned}
    \includegraphics[width=0.45\linewidth]{SlideV.pdf}
    \end{aligned}  
\end{equation}
This identity can be seen by observing that $[\tilde{Z}^q\otimes \id \otimes \omega_{h(q)}]\otimes \id\otimes\id\ket{\psi}=\id\otimes [\tilde{Z}^q\otimes \omega^T_{h(q)}\otimes \id]\otimes\id\ket{\psi}$, where $\ket{\psi}$ is the representative state depicted in Eq.~\eqref{eq:FiducialRep} in the main text. By inserting an identity, $\id = [\id\otimes \omega^T_{h(q)}\otimes \omega^\dagger_{h(q)}][\id\otimes \omega^*_{h(q)}\otimes \omega_{h(q)}]$, one notices that $\id\otimes [\tilde{Z}^q\otimes \omega^T_{h(q)}\otimes \id]\otimes \id\ket{\psi}=\id\otimes [\id\otimes \omega^T_{h(q)}\otimes \omega^\dagger_{h(q)}][\tilde{Z}^q\otimes \id\otimes \omega_{h(q)}]\otimes \id\ket{\psi}$. Finally, by noticing that
\begin{align}
    &\id \otimes \omega_{h(q)}^T\otimes \omega_{h(q)}^\dagger = \left(\id \otimes \id \otimes \id\right) \left[\bigoplus_{\alpha\in K} W_{h(q)}^\dagger \right]\\
    & = \left(P_{k(|K| h(q))}^T \otimes \id \otimes \id\right) \left[\bigoplus_{\alpha\in K} W_{h\left(|K|h(q)\oplus_G \alpha\right)}^\dagger \right]\\
    &  = U_{|K|h(q)}^\dagger,
\end{align}
where again we have used the group identities, the identity directly follows.

\subsection{Proof of Lemma \ref{lem:ZeroVectorOutcome}} \label{app:ZeroVectorOutcome}
Finally, we also prove Lemma~\ref{lem:ZeroVectorOutcome}. Let $\tilde{U}^{(0)} = \prod_{j=1}^n U_{r_j}$ as explained in the main text.
\begin{manuallemma}{4}
    Let $\tilde{U}^{(0)} \ne \id$. The following equation holds.
    \begin{equation}\label{eq:ZeroVectorOutcomeApp}
    \begin{aligned}
        \includegraphics[width=0.5\linewidth]{ZeroVectorOutcome.pdf}
    \end{aligned}
    \end{equation}
\end{manuallemma}

\begin{proof}
    The left-hand side of Eq.~\eqref{eq:ZeroVectorOutcomeApp} can be written equivalently as
    \begin{multline}
        \ket{\psi} =\sum_{\vec{i},\vec{j},\vec{k}} \tr[U_g  A^{i_1,j_1,k_1}A^{i_2,j_2,k_2}\dots A^{i_n,j_n,k_n}]\\
        \ket{i_1 j_1 k_1 i_2 j_2 k_2\dots i_n j_n k_n},
    \end{multline}
    for some $g\in G$. Consider the coefficients for each $\vec{i},\vec{j},\vec{k}$. We have
    \begin{align}
       &  \tr[U_g  A^{i_1,j_1,k_1}A^{i_2,j_2,k_2}\dots A^{i_n,j_n,k_n}]\nonumber\\
       &= \sum_{\vec{\gamma}} \tr \Big[ U_g  \ketbra{i_1 j_1 \gamma_1}{i_1 k_1 \gamma_1}\ketbra{i_2 j_2 \gamma_2}{i_2 k_2 \gamma_2}\dots \nonumber\\
       &\qquad \qquad\qquad \qquad\qquad \qquad\ketbra{i_n j_n \gamma_n}{i_n k_n \gamma_n}  \Big]\\
       &= \delta_{i_2,i_1}\delta_{i_3,i_2}\dots \delta_{i_{n},i_{n-1}}\delta_{j_2,k_1}\delta_{j_3,k_2}\dots \delta_{j_{n},k_{n-1}}\nonumber\\
       &\qquad \qquad\qquad \qquad\qquad \sum_{\gamma_1 } \bra{i_n k_n \gamma_1} U_g \ket{i_1 j_1 \gamma_1}\\
       &= \delta_{i_2,i_1}\delta_{i_3,i_2}\dots \delta_{i_{n},i_{n-1}}\delta_{j_2,k_1}\delta_{j_2,k_1}\dots \delta_{j_{n},k_{n-1}}\nonumber\\
       &\qquad \qquad\qquad \qquad\qquad \sum_{\gamma_1 } \bra{i_1 k_n \gamma_1} U_g \ket{i_1 j_1 \gamma_1}.
    \end{align}
    Now taking the last term, we have
    \begin{align}
         &\sum_{\gamma_1} \bra{i_1 k_n \gamma_1} U_g \ket{i_1 j_1 \gamma_1}\nonumber \\
         &\qquad = \sum_{ \alpha \in K} \bra{i_1} \ketbra{\alpha}{\alpha} X^{k(g)} \ket{i_1} \nonumber \\
         & \qquad \qquad \quad \bra{k_n} \omega_{h(g\oplus \alpha)}^* \ket{j_1}  \sum_{\gamma_1} \bra{\gamma_1} \omega_{h(g \oplus \alpha)} \ket{\gamma_1}\\
         & =  \bra{i_1} X^{k(g)} \ket{i_1} \bra{k_n} \omega_{h(g\oplus i_1)}^* \ket{j_1} \tr[\omega_{h(g \oplus i_1)}]\\
         & = D_\mu \delta_{k(g),0} \delta_{h(g\oplus i_1)}  \bra{k_n} \omega_0^* \ket{j_1}\\
         & = D_\mu^2 \delta_{k(g),0} \delta_{h(g)}  \delta_{j_1,k_n}\\
         & = \delta_{g,0} \delta_{j_1,k_n} D_\mu^2,
    \end{align}
    where we have used the group identities to reach the second to last line. Thus, Eq.~\eqref{eq:ZeroVectorOutcomeApp} is zero unless $g=0$, i.e., unless $U_g=\id$.
\end{proof}

\section{Non-Abelian Symmetries}\label{app:non-abelian}
In this appendix, we provide the supporting proofs for the claims made in Section~\ref{sec:NonAbelianSymmetries}.

\subsection{Proof of Result~\ref{res:LOCC}}\label{app:proofLOCC}

In this section, we prove Result~\ref{res:LOCC}. Any linear unitary representation of a non-abelian group can be split into (1D) abelian irreps and $(D>1)$ non-abelian irreps. That is, we may write $U_g\cong \equiv U_g^{A} \oplus U_g^{NA}$. This decomposition also decomposes the Hilbert space on which $U_g$ acts, $\mathcal{H}\cong\mathcal{H}^{A} \oplus \mathcal{H}^{NA}$. By Schur's lemma, symmetric reduced density matrices must also decompose this way, $\rho \cong \rho^{A} \oplus \rho^{NA}$. Now, for convenience, let us restate the theorem. 
\begin{manualresult}{5}
Let $G$ be a non-abelian group such that $\omega_g^\mu \otimes (\omega^{\mu}_g)^*$ is non-abelian for all nontrivial projective $\mu$-irreps $\omega_g^\mu$. Then, for any translationally invariant SPT state $\ket{\psi_n}$ associated with $\mu$, there is some $n_0$ such that $\forall n>n_0$ $\ket{\psi_n}$ cannot be deterministically converted via symmetric LOCC to a product state.
\end{manualresult}
\begin{proof}
        Recall that SPT states correspond to normal MPS, and therefore the tensor, after blocking and in canonical form, has only one block. Consequently, as discussed in Appendix~\ref{app:noninjectivesym}, after blocking sufficiently many sites, the physical symmetry in an appropriate basis decomposes as $(\omega_g\otimes\omega_g^*)\oplus U_g^{\rm extra}$, where $\omega_g$ is a projective representation corresponding to a nontrivial cohomology class $\mu$. At this length scale and in this decomposition, the reduced density matrix of the MPS has the form $\rho \oplus 0$, where $\rho$ is full rank \cite{PerezGarciaEtAl2007_MPSrepresentations}. In particular, as $\omega_g\otimes \omega_g^*$ is non-abelian, $\rho$ has support on a non-abelian subspace.
        
        Now, let $n$ be big enough that one can bipartition the state into two regions, $A$ and $B$, large enough that the above properties hold. Considered then as a bipartite system, we have
    \begin{equation}
        U_g^A \otimes U_g^B \ket{\psi}_{AB} \propto \ket{\psi}_{AB}
    \end{equation}
    for some non-abelian representations $U_g^{A,B}$, and the reduced density matrices of $\ket{\psi}_{AB}$ on $A$ and $B$ have support on the non-abelian subspaces. Note, that the fact $\rho_A$ has non-abelian support is alone sufficient to ensure that $\ket{\psi}_{AB}$ is entangled as if $\tr[\rho^{NA}]\ne 0$, then $\text{rk}(\rho)>1$.

    Let us consider a symmetric projective measurement applied by Alice. As the measurement is symmetric, by Schur's lemma and the measurement completeness, there must be some outcome that has support on the non-abelian subspace. Moreover, again by Schur's lemma, the reduced density matrix cannot be rank 1. Thus at least one outcome remains entangled.
    
    Let us now consider all the possible allowed operations. 
    As any auxiliary system must be initialized in a pure eigenstate - therefore having support on the abelian subspace - adding auxiliary systems cannot eliminate the states support on the non-abelian subspace. Moreover, as with Ref.~\cite{LoPopescu2001_BipartiteLOCC} (see also Refs~\cite{SchuchEtAl2004_SuperselectionLOCC1,SchuchEtAl2004_SuperselectionLOCC2}), any symmetric projective measurement by $B$ can be simulated by a projective measurement by $A$ (with symmetric local unitary corrections on $B$). Therefore, no actions by either party can deterministically eliminate the support on the non-abelian subspace.
    Finally, by the assumption that $\omega_g^\mu \otimes (\omega^{\mu}_g)^*$ is non-abelian for all nontrivial projective $\mu$-irreps $\omega_g^\mu$, this holds for all SPT states in the phases corresponding to the cohomology class $\mu$.
\end{proof}

\ \\

\subsection{Proof of Eq.~\eqref{eq:probDistr1}}\label{app:evenLemma}
Starting from the expression
\begin{equation}
    \bra*{{\rm SPT}_\abs{f}} P_f^{\otimes \abs{f}} \ket*{{\rm SPT}_\abs{f}},
\end{equation}
one uses the fact that, in general, $X_{AB}\otimes \id_C\ket*{\Phi^+}=X_{AC}^{T_2}\otimes \ket{\Phi^+}$, where $T_2$ denotes the partial transpose in the computational basis on the second system. This leads to the expression
\begin{equation}
    \bra*{{\rm SPT}_\abs{f}} P_f^{\otimes \abs{f}} \ket*{{\rm SPT}_\abs{f}}=\bra*{{\rm SPT}_\abs{f}}\prod_{i\in \abs{f}} P_f^{(i,i+1)}\ket*{{\rm SPT}_\abs{f}},
\end{equation}
where, like $\ket*{{\rm SPT}_\abs{f}}$, $P_f^{(i,i+1)}$ should be understood to have periodic boundary conditions, and where we have used that $P_f$ is invariant under partial transposition as it is diagonal in the computational basis. From there it is easy to see that the operator product $\prod_{i\in \abs{f}} P_f^{(i,i+1)}=0$ iff $\abs{f}$ is odd. We may graphically represent the previously discussed steps for $\abs{f}=3$ as
\begin{equation}
\begin{aligned}
    \includegraphics[width=0.5\linewidth]{P3network.pdf}
\end{aligned}=0.
\end{equation}
In case $\abs{f}$ is even the expression above yields $\frac{2}{2^{\abs{f}}}$, which proves Eq.~\eqref{eq:probDistr1}.

\subsection{Probability of obtaining Non-Correctable Error States}\label{app:probError}
Finally, here we expand on the argument for why, in the limit $n\rightarrow\infty$, a $\text{log}(n)$-depth circuit ensures an asymptotically deterministic transformation. Let us assume we are given a certain distribution of error outcomes in the output string $\vec{x}$, and we can apply a circuit of a certain depth $l$. Then one needs a circuit of at least depth $l+1$ to correct $\vec{x}$ if and only if there is an error at some $k_i$ with no other error outcome within distance $\pm 2l$ to it. Such an output string is not correctable and the protocol would fail. As derived in the main text, the probability distribution is flat on outcomes with even number of $f$s and zero otherwise. Therefore, the probability $p_{\rm fail}$ to obtain a non-correctable outcome $\vec{x}$ is given by $p_{\rm fail}=(\# \vec{x}\mathrm{\,not\,correctable})/2^{n-1}$. The number of error strings can be upper bounded by fixing an error at position $i$ and no errors withing distance $\pm 2l$. This is given by $2^{n-(4l+1)-1}$. There are $n$ possibilities for $i$, hence we can upper bound
\begin{equation}
    p_{\rm fail}\leq \frac{n 2^{n-(4l+1)-1}}{2^{n-1}}=\frac{n}{2^{4l+1}}.
\end{equation}
If $l=l(n)\equiv\log(n)$ one obtains $p_{\rm fail}\rightarrow 0$, for $n\rightarrow \infty$.

\subsection{Proof of Eq.~\eqref{eq:probDistr}}\label{app:GHZ}
From the relations in Eqs.~\eqref{eq:relation1}, and~\eqref{eq:relation2} in the main text one can deduce that upon the first round of measurements, the probability of obtaining a completely successful outcomes, i.e., $\vec{x}=(s,s,\dots,s)$, is given by
\begin{align}
    p(\vec{x})&=\sum_{\vec{i}\in\{0,\dots ,3\}^{n}} \abs*{\bra{\phi_0}^{\otimes n}\bigotimes_{k\in [n]} (\tilde{Z}^{i_k}_k)^\dagger\ket*{{\rm GHZ}_n^8}}^2\\
    &=\left(\frac{1}{8}\right)^{n-1}\sum_{\vec{i}\in\{0,\dots ,3\}^{n}}\abs*{\bra{\varphi_0}\prod_{k\in [n]} (\tilde{Z}^{i_k})^\dagger\ket{\varphi_0}}^2\\
    &=\frac{1}{2^{n-1}}
\end{align}
where the second equation follows from Eqs.~\eqref{eq:relation1}, and~\eqref{eq:relation2} in the main text. The last equation follows from the fact that the overlap is one in the case $\prod_{k\in [n]} (\tilde{Z}^{i_k})^\dagger=\id$, for which there are $4^{n-1}$ possibilities, and zero otherwise. Moreover, any post-measurement state is a product state and, up to quasi-commuting unitaries, equivalent to the state $\ket{\varphi_0}^{\otimes n}$.

Whenever unsuccessful outcomes occur one can verify that the probabilities of obtaining a certain distribution of successes can at most depend on the number of successes, and not their distribution. To see this, consider a $p(\vec{x})$ with $\abs{f}$ unsuccessful outcomes; it reads
\begin{equation}
   p(\vec{x}) = 4^{n-\abs{f}}\qty(\frac{1}{8})^{n-\abs{f}}\bra*{{\rm GHZ}^8_{\abs{f}}}P_f^{\otimes \abs{f}}\ket*{{\rm GHZ}^8_{\abs{f}}}.
\end{equation}
Upon inspection one notices that $\ket*{{\rm GHZ}^8_{\abs{f}}}\simeq \ket*{{\rm GHZ}_{\abs{f}}^{2}}\otimes \ket*{{\rm GHZ}_{\abs{f}}^{2}}\otimes \ket*{{\rm GHZ}_{\abs{f}}^{2}}$, and $P_f=\id_2\otimes\ketbra{-}\otimes\id_2$. It follows that the probability above evaluates to zero whenever $\abs{f}$ is odd, and otherwise yields a factor $1/2^{\abs{f}-1}$. This proves that the probability distribution $p(\vec{x})$ is flat for $\abs{f}$ even, and zero otherwise, which we claimed in Eq.~\eqref{eq:probDistr} in the main text.

\end{document}